\documentclass[12pt,onecolumn]{IEEEtran}

\usepackage{epsf}
\usepackage{graphicx}
\usepackage{amsmath} \usepackage{amssymb}
\usepackage{booktabs}
\usepackage{subfigure}
\usepackage{cases}

\newenvironment{proof}{\begin{IEEEproof}}{\end{IEEEproof}}

\newtheorem{theorem}{Theorem}[section]

\newtheorem{proposition}{Proposition}[section]
\newtheorem{lemma}{Lemma}[section]
\newtheorem{corollary}{Corollary}[section]

\long\def\symbolfootnote[#1]#2{\begingroup%
\def\thefootnote{\fnsymbol{footnote}}\footnote[#1]{#2}\endgroup}

\def\dref#1{(\ref{#1})}

\def\be{\begin{equation}} \def\ee{\end{equation}}
\def\ba{\begin{array}} \def\ea{\end{array}} \def\bna{\begin{eqnarray}}
\def\ena{\end{eqnarray}}

 \def\bna{\begin{eqnarray}}
\def\ena{\end{eqnarray}} \def\dref#1{(\ref{#1})}
\begin{document}

\title{Cut-Set Bound Is Loose for \\
Gaussian Relay Networks}

\author{Xiugang Wu and Ayfer \"{O}zg\"{u}r

\thanks{This work was supported in part by the NSF CAREER award 1254786, NSF award CCF-1514538 and by the Center for Science of Information
(CSoI), an NSF Science and Technology Center, under grant agreement CCF-0939370. This paper was presented in part at the 2015 Allerton Conference on Communication, Control, and Computing \cite{Allerton2015} and the 2016 IEEE International Symposium on Information Theory \cite{ISIT2016}.}
\thanks{X. Wu and A. \"{O}zg\"{u}r are with the Department of Electrical Engineering, Stanford University, Stanford, CA 94305, USA (e-mail: x23wu@stanford.edu; aozgur@stanford.edu).}

}

\maketitle

\begin{abstract}
The cut-set bound developed by Cover and El Gamal in 1979 has since remained the best known upper bound on the capacity of the Gaussian relay channel. We develop a new upper bound on the capacity of the Gaussian primitive relay channel which is tighter than the cut-set bound. Our proof is based on typicality arguments and concentration of Gaussian measure. Combined with a simple tensorization argument proposed by Courtade and Ozgur in 2015, our result also implies that  the current capacity approximations for Gaussian relay networks, which have linear gap to the cut-set bound in the number of nodes, are order-optimal and leads to a lower bound on the pre-constant.
\end{abstract}

\section{Introduction}\label{S:Introduction}



The single-relay channel is one of the simplest examples of a network information theory problem, which defies our complete understanding despite decades of research. The Gaussian version of this problem models the communication scenario where a wireless link is assisted by a single relay. Motivated by the need to increase the spectral efficiency of wireless systems and the increasing importance of relaying for small cells, it has been studied extensively since its formulation in 1971 \cite{van71}. However, the characterization of its capacity still remains an open problem. Perhaps more interestingly, the existing literature almost exclusively focuses on developing achievable strategies for this channel as well as larger relay networks. This has led to a plethora of relaying schemes  over the last decade, such as decode-and-forward, compress-and-forward, amplify-and-forward, compute-and-forward, quantize-map-and-forward, noisy network coding, etc \cite{covelg79}--\cite{Limetal}. In sharp contrast, the only available upper bound on the capacity of the Gaussian relay channel is the so called cut-set bound developed by Cover and El Gamal in 1979 \cite{covelg79}. In the 40-year long literature on the problem, the cut-set bound has been consistently used as a benchmark for performance --for example the recent approximation approach \cite{Avestimehretal,OzgurDiggavi,Limetal} in wireless information theory focuses on bounding the gap of the achievable
strategies to the cut-set bound of the network-- however to our knowledge, whether the cut-set bound is indeed achievable or not in a Gaussian relay channel (except in trivial cases) remains unknown to date.

\begin{figure}[t!]
\centering
\includegraphics[width=0.45\textwidth]{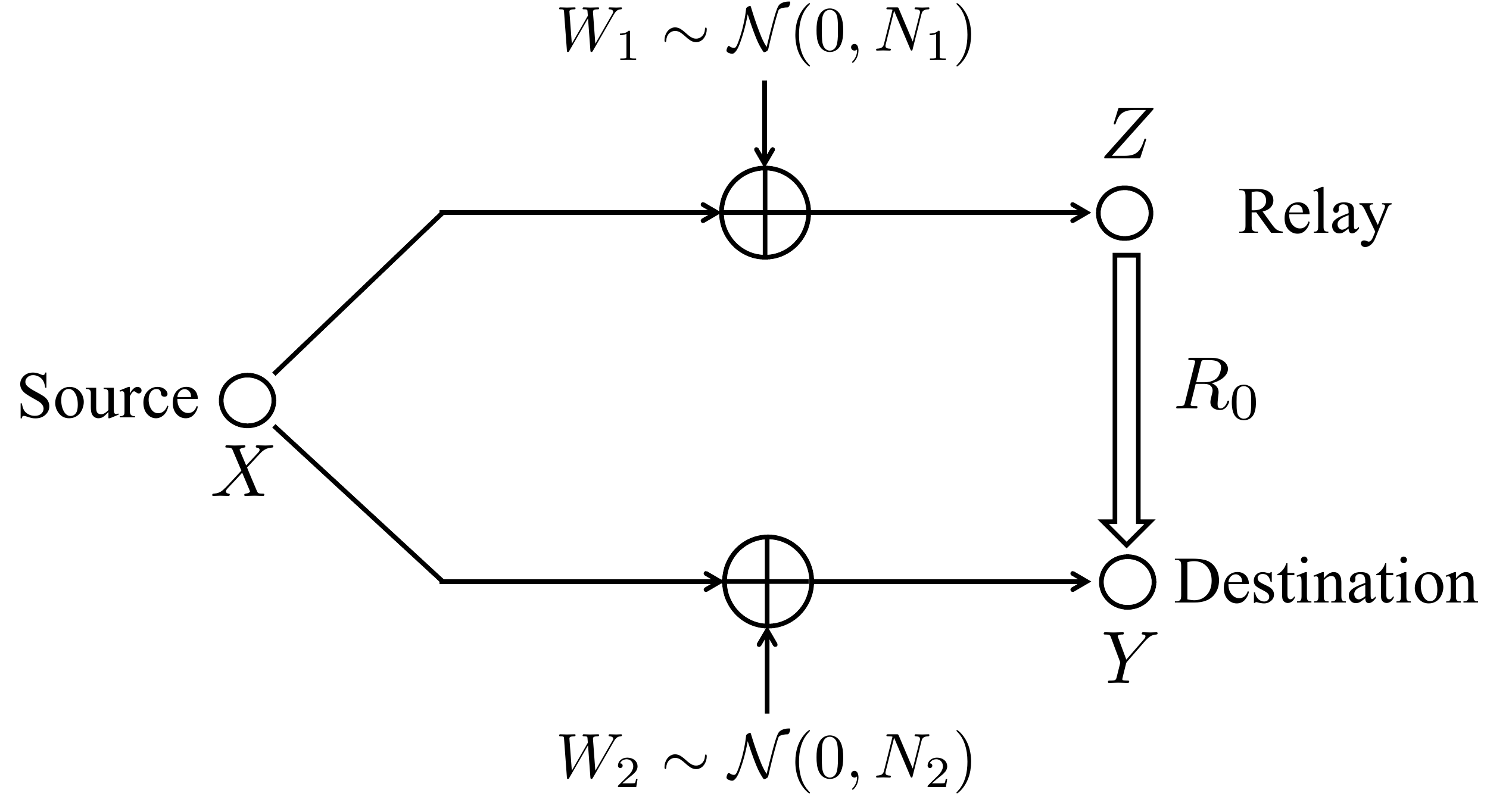}
\caption{Gaussian primitive relay channel.}
\label{F:GaussianRelay}
\end{figure}

In this paper, we make progress on this problem by developing a new upper bound on the capacity of the Gaussian primitive relay channel. This is a special case of the Gaussian single relay channel where the multiple-access channel from the source and the relay to the destination has orthogonal components \cite{KimAllerton}. See Figure~\ref{F:GaussianRelay}. Here, the relay can be thought of as communicating to the destination over a Gaussian channel in a separate frequency band, or equivalently the destination can be thought of as equipped with two receive antennas, one directed to the source and one directed to the relay with no interference in between.\footnote{Note that due to  network equivalence, the rate limited channel from the relay to the destination in Figure~\ref{F:GaussianRelay} can be equivalently thought of  as a Gaussian channel of the same capacity \cite{Koetteretal}.} Our upper bound is tighter than the cut-set bound for this channel for all (non-trivial) channel parameters. While this result is developed for the single-relay setting, it  has implications also for networks with multiple relays. In particular, combined with a simple tensorization argument recently proposed in \cite{CourtadeOzgur}, it implies that the linear (in the number of nodes) gap to the cut-set bound in  current capacity approximations for Gaussian relay networks is indeed fundamental. The capacity of Gaussian relay networks can have linear gap to the cut-set bound and our result can be used to obtain a lower bound on the pre-constant.

Proving the above result requires to capture the following phenomenon: if a relay is not
able to decode the transmitted message and therefore remove the noise in its received signal by decoding, then the signal it forwards necessarily contains noise along with information. The injected noise then decreases the end-to-end achievable rate with respect to the cut-set bound, where the latter simply upper bounds the end-to-end capacity by the maximal information  flow over cuts of the network assuming all nodes on the source side of the cut have noiseless access to the message and all nodes on the destination side can freely cooperate to decode the transmitted message. As basic as it sounds, existing approaches for developing converses in information theory seem insufficient to quantitatively capture this phenomenon.

In this and our concurrent work \cite{WuXieOzgur_ISIT2015}--\cite{WuOzgurXie_TIT} on the discrete memoryless version of this problem, we build a novel geometric approach to capture these tensions. We use measure concentration to study the probabilistic geometric relations between the typical sets of the $n$-letter random variables associated with the problem. We then translate these geometric relations between typical sets into new and surprising relations between the entropies of the corresponding random variables. While our bounds for the discrete memoryless relay channel in \cite{WuXieOzgur_ISIT2015}--\cite{WuOzgurXie_TIT} and the Gaussian case treated in the current paper have similar flavor, these two cases also comprise some significant differences. In particular, the discrete memoryless case seems easier to deal with as one can do explicit counting arguments and rely on the standard notion of strong typicality. For example, earlier upper bounds in \cite{Zhang}--\cite{Xue} for the discrete memoryless relay channel rely on such counting arguments and cannot be extended to the Gaussian case. 

\subsection{Organization of the Paper}
The remainder of the paper is organized as follows. First Section \ref{S:ChannelModel} introduces the channel model and reviews the classical cut-set bound on the capacity of the Gaussian primitive relay channel.  Then Section \ref{S:mainresults} presents our new upper bound and discusses its implication on the capacity approximation problem for Gaussian relay networks, followed by the proof of our bound in Section \ref{S:symproof}. Finally in Section \ref{S:furtherimprove}, we provide another bound which sharpens our main result for certain regimes of the channel parameters. One of our motivations to include this result is to illustrate that there may be significant potential for improving our results by refining our method and arguments.


%
%
%
%

\section{Preliminaries}\label{S:ChannelModel}

\subsection{Channel Model}
Consider a Gaussian primitive relay channel as depicted in Fig. \ref{F:GaussianRelay}, where $X\in\mathbb{R}$ denotes the source signal which is constrained to average power $P$, and $Z\in\mathbb{R}$ and $Y\in\mathbb{R}$ denote the received signals of the relay and the destination. We have
\begin{numcases}{}
Z=X+W_1\nonumber \\
Y=X+W_2 \nonumber
\end{numcases}
where $W_1$ and $W_2$ are Gaussian noises that are independent of each other and $X$, and have zero mean and variances $N_1$ and $N_2$ respectively. The relay can communicate to the destination via an error-free digital link of rate $R_0$.

For this channel, a code of rate $R$ and blocklength $n$, denoted by $$(\mathcal{C}_{(n,R)}, f_n(z^n), g_n(y^n,f_n(z^n))), \mbox{ or simply, } (\mathcal{C}_{(n,R)}, f_n, g_n), $$
consists of the following:
\begin{enumerate}
  \item A codebook at the source $X$,
$$\mathcal{C}_{(n,R)}=\{x^n (m), m\in \{1,2,\ldots, 2^{nR}\} \}$$
where $$\frac{1}{n}\sum_{i=1}^n  x_i^2 (m)\leq P, \ \forall m \in \{1,2,\ldots,2^{nR}\};$$
  \item An encoding function at the relay $Z$,
$$f_n: \mathbb R^n \rightarrow \{1,2,\ldots, 2^{nR_0}\};$$
  \item A decoding function at the destination $Y$,
$$g_n: \mathbb R^n  \times \{1,2,\ldots, 2^{nR_0}\}  \rightarrow \{1,2,\ldots, 2^{nR}\}.$$
\end{enumerate}

The average probability of error of the code is defined as
$$P_e^{(n)}=\mbox{Pr}(g_n(Y^n,f_n(Z^n)) \neq M ),$$
where the message $M$ is assumed to be uniformly drawn from the message set $ \{1,2,\ldots, 2^{nR}\}$. A rate $R$ is said to be achievable if there exists a sequence of codes
$$\{(\mathcal{C}_{(n,R)}, f_n, g_n)\}_{n=1}^{\infty}$$
such that the average probability of error $P_e^{(n)} \to 0$ as $n \to \infty$.
The capacity of the primitive relay channel is the supremum of all achievable rates, denoted by $C(R_0)$.

\subsection{The Cut-Set Bound}

For the Gaussian primitive relay channel, the cut-set bound can be stated as follows.

\begin{proposition}[Cut-set Bound]\label{P:cutset}
For the Gaussian primitive relay channel, if a rate $R$ is achievable, then there exists a random variable $X$ satisfying $E[X^2]\leq P$ such that
\begin{numcases}{}
 R   \leq I(X;Y,Z)\label{E: cut1} \\
R    \leq   I(X;Y)+R_0  \label{E: cut2}.
\end{numcases}
\end{proposition}

Note that  constraints \dref{E: cut1} and \dref{E: cut2}  correspond to the broadcast channel $X$-$YZ$ and multiple-access  channel $XZ$-$Y$, and hence are generally known as the broadcast and multiple-access constraints, respectively. Also it can be easily shown (c.f. Appendix \ref{Gaussianoptimal}) that both $I(X;Y,Z)$ and $I(X;Y)$ in Proposition \ref{P:cutset} are maximized when $X\sim\mathcal N (0,P)$, which leads us to the following corollary.

\begin{corollary}\label{C:cutset}
For the  Gaussian primitive relay channel, if a rate $R$ is achievable, then
\begin{numcases}{}
 R   \leq \frac{1}{2}\log \left(1+\frac{P}{N_1}+\frac{P}{N_2}\right) \label{E:cutsetgeneral1} \\
R    \leq   \frac{1}{2}\log \left(1+\frac{P}{N_2}\right)+R_0  .  \label{E:cutsetgeneral2}
\end{numcases}
\end{corollary}



\section{Main Result}\label{S:mainresults}

Our main result in this paper is the following theorem, which provides a new upper bound on the capacity of the Gaussian primitive relay channel that is tighter than the cut-set bound. The proof of this theorem is given in Section \ref{S:symproof}.
\begin{theorem}\label{T:newboundsym}
For the Gaussian primitive relay channel, if a rate $R$ is achievable, then there exists some $a\in [0,R_0]$ such that
\begin{numcases}{}
 R   \leq \frac{1}{2}\log \left(1+\frac{P}{N_1}+\frac{P}{N_2}\right)\label{E:newboundsym1}  \\
R    \leq   \frac{1}{2}\log \left(1+\frac{P}{N_2}\right)+R_0  -a  \label{E:newboundsym2} \\
R    \leq   \frac{1}{2}\log \left(1+\max \left\{\frac{P}{N_1},\frac{P}{N_2}\right\}\right)+a+\sqrt{2a\ln2}\log e.   \label{E:newboundsym3}
\end{numcases}
\end{theorem}

Since $a\geq 0$ in the above theorem, our bound is in general tighter than the cut-set bound in Corollary \ref{C:cutset}. In fact,  our bound can be  \emph{strictly} tighter than the cut-set bound when the multiple-access constraint \eqref{E:cutsetgeneral2} is active in the cut-set bound. To see this, first consider the symmetric case when $N_1=N_2=:N$. For this case, the cut-set bound in Corollary \ref{C:cutset} says that if a rate $R$ is achievable, then
\begin{numcases}{}
 R   \leq \frac{1}{2}\log \left(1+\frac{2P}{N}\right)  \label{E:cutsetsym1} \\
R    \leq   \frac{1}{2}\log \left(1+\frac{P}{N}\right)+R_0,\label{E:cutsetsym2}
\end{numcases}
while our bound in Theorem~\ref{T:newboundsym} asserts that any achievable rate $R$ must satisfy
\begin{numcases}{}
R   \leq \frac{1}{2}\log \left(1+\frac{2P}{N}\right)  \label{E:gapcom1} \\
R    \leq   \frac{1}{2}\log \left(1+\frac{P}{N}\right)+R_0-a^*, \label{E:gapcom2}
\end{numcases}
where $a^*$ is the solution to the following equation:
\begin{align}
R_0=2a^*+ \sqrt{2a^*\ln2} \log e, \label{E:solvegap}
\end{align}
which is obtained by equating the R.H.S. of constraints \eqref{E:newboundsym2} and \eqref{E:newboundsym3}. Obviously, if $R_0>0$, then $a^*>0$ and \eqref{E:gapcom2} is tighter than \eqref{E:cutsetsym2}. Therefore, when constraint \dref{E:cutsetsym2} is more stringent between \dref{E:cutsetsym1} and \dref{E:cutsetsym2}, our bound is strictly tighter than the cut-set bound. The same argument and conclusion also apply when $N_1\geq N_2$, in which case our bound reduces to
\begin{numcases}{}
R   \leq \frac{1}{2}\log \left(1+\frac{P}{N_1}+\frac{P}{N_2}\right)   \\
R    \leq   \frac{1}{2}\log \left(1+\frac{P}{N_2}\right)+R_0-a^*,
\end{numcases}
where $a^*$ is similarly defined as in \dref{E:solvegap}. Finally it can be easily checked that when $N_1 \leq N_2$, our bound is also strictly tighter than the cut-set bound as long as
$$\frac{1}{2}\log \left(1+\frac{P}{N_1}\right)\leq \frac{1}{2}\log \left(1+\frac{P}{N_2}\right)+R_0.$$

Note that both the cut-set bound and our bound depend on the channel parameters through $\frac{P}{N_1}, \frac{P}{N_2}$ and $R_0$. It is interesting to evaluate the largest gap between these two bounds over all parameter values $(\frac{P}{N_1}, \frac{P}{N_2},R_0)$. For this we show in Appendix \ref{A:largestgap} the following proposition, which says that the largest gap occurs  in the symmetric case when $\frac{P}{N_1}= \frac{P}{N_2}\rightarrow \infty$ and $R_0=0.5$.

\begin{proposition}\label{P:largestpossible}
 Let $\Delta\left(\frac{P}{N_1},\frac{P}{N_2},R_0\right)$ denote the gap between our bound and the cut-set bound, and $\Delta^*$ be its largest possible  value over all Gaussian primitive relay channels, i.e.,
$$\Delta^*:=\sup_{\frac{P}{N_1},\frac{P}{N_2}, R_0} \Delta\left(\frac{P}{N_1},\frac{P}{N_2},R_0\right).$$
Then, $\Delta^* =  \Delta(\infty,\infty, 0.5)=0.0535$.
\end{proposition}




\subsection{Gaussian Relay Networks}\label{S:lineargap}

The primitive single-relay channel we consider in this paper can be regarded as a special case of a Gaussian relay network. However, the upper bound we develop for this special case has also implications for larger Gaussian relay networks with multiple relays. In particular, it can be used to infer how tightly the capacity of general Gaussian relay networks can be approximated by the cut-set bound. Initiated by the work of Avestimehr, Diggavi and Tse \cite{Avestimehretal}, there has been significant recent interest \cite{OzgurDiggavi,Limetal} in approximating the capacity of general Gaussian relay networks with the cut-set bound, i.e. bounding the gap between the rates achieved by specific schemes and the cut-set bound on capacity. The gap in these approximation results is linear in the number of nodes in the network but independent of the channel SNRs and network topology. In particular, the best currently known approximation result \cite{DDF} has a gap of $0.5N$ where $N$ is the total number of nodes.

However, an approximation gap that increases linearly in
the total number of nodes quickly becomes too large even
for networks of moderate size. Therefore an interesting question, posed as an open problem in \cite{FT}, is
whether this linear gap can be substantially improved, for
example, to scale sublinearly in the total number of nodes. Some recent results \cite{Urs1,Urs2,Bobbie,Ritesh} encourage this possibility by  demonstrating that sublinear in the number of nodes (or in the total number of antennas in the case of multiple antenna nodes) gap to the cut-set bound can be achieved when additional constraints are imposed on the topology of the network. However, a more recent tensorization argument proposed in \cite{CourtadeOzgur} shows that the gap between the capacity and the cut-set bound can be bounded by a sublinear function of the number of nodes, independent of network topology
and channel configurations, if, and only if, capacity is equal to
the cut-set bound for \emph{all} Gaussian relay networks. 
Moreover, Theorem~3 of \cite{CourtadeOzgur} shows that an explicit gap to the cut-set bound for any specific network with specific channel parameters and topology can be used to obtain a lower bound on the pre-constant in these approximation results. Therefore, the gap $0.0535$ in Proposition~\ref{P:largestpossible} for the Gaussian primitive relay channel implies that the capacity of Gaussian relay networks can not be approximated by the cut-set bound, independent of the topology and the channel coefficients, with a gap that is smaller than $(0.0535/4)N\approx 0.01 N$. Note that the primitive relay channel can be thought of as a Gaussian network with two receive antennas at the destination, one directed to the source and one directed to the relay with no interference in
between, so this network can be thought of as a Gaussian relay network comprised of four antennas in total.

%


\section{Proof of Theorem~\ref{T:newboundsym}}\label{S:symproof}
In this section we will first provide a proof of Theorem~\ref{T:newboundsym}  for the symmetric case  and then generalize it to the asymmetric case. We begin by observing that the symmetric case of Theorem~\ref{T:newboundsym} follows as a corollary to the following proposition.

\begin{proposition} \label{P:newboundsym}
For the symmetric Gaussian primitive relay channel, if a rate $R$ is achievable, then there exists a random variable $X$ satisfying $E[X^2]\leq P$ and some $a\in [0,R_0]$ such that
\begin{numcases}{}
 R   \leq I(X;Y,Z)\label{E: symnew1} \\
R    \leq   I(X;Y)+R_0  -a \label{E: symnew2}\\
R    \leq   I(X;Y)+a+\sqrt{2a\ln 2}\log e . \label{E: symnew3}
\end{numcases}
\end{proposition}

To see Proposition \ref{P:newboundsym} implies Theorem~\ref{T:newboundsym} when $N_1=N_2$, simply observe that as in the case of the cut-set bound, all the mutual information terms in Proposition \ref{P:newboundsym} are maximized when $X\sim\mathcal N (0,P)$. Therefore, to prove Theorem~\ref{T:newboundsym} for the symmetric case, it suffices to prove Proposition \ref{P:newboundsym} and we will do this by proving constraints \dref{E: symnew1}--\dref{E: symnew3} sequentially with the main step being the proof of \dref{E: symnew3}.

\bigbreak

\text{Proof of Proposition~\ref{P:newboundsym}:} Suppose a rate $R$ is achievable. Then there exists a sequence of codes
\begin{align}
\{(\mathcal{C}_{(n,R)}, f_n, g_n)\}_{n=1}^{\infty}  \label{E:ReliableCodes}
\end{align}
such that the average probability of error $P_e^{(n)} \to 0$ as $n \to \infty$.

For this sequence of codes, we have
\begin{align}
n R &= H(M)\nonumber \\
&=I(M;Y^n,Z^n)+H(M|Y^n,Z^n)\nonumber \\
&\leq I(X^n;Y^n,Z^n)+H(M|Y^n,f_n(Z^n))\nonumber \\
&\leq I (X^n;Y^n,Z^n)+n \mu  \label{E:GaussianFano} \\
&= h(Y^n,Z^n)- h(Y^n,Z^n|X^n)+n  \mu   \nonumber \\
&= \sum_{i=1}^{n} [h(Y_i,Z_i|Y^{i-1},Z^{i-1})- h(Y_i,Z_i|X_i)] +n \mu  \nonumber \\
&\leq  \sum_{i=1}^{n} [h(Y_i,Z_i)- h(Y_i,Z_i|X_i)] +n \mu  \nonumber \\
&=  \sum_{i=1}^{n} I (X_i;Y_i,Z_i) +n \mu   \nonumber \\
&= n(I(X_Q;Y_Q,Z_Q|Q)+ \mu ) \label{E:Timesharing}\\
&= n(h(Y_Q,Z_Q|Q)- h(Y_Q,Z_Q|Q,X_Q)  +\mu ) \nonumber\\
&\leq  n(h(Y_Q,Z_Q)- h(Y_Q,Z_Q|X_Q)  + \mu ) \nonumber\\
&= n(I(X_Q;Y_Q,Z_Q)+ \mu ) \nonumber
\end{align}
i.e.,
\begin{align}
R\leq  I(X_Q;Y_Q,Z_Q)+ \mu  \label{E:summarize1}
\end{align}
for any $\mu  >0$ and sufficiently large $n$, where \dref{E:GaussianFano} follows from Fano's inequality,  \dref{E:Timesharing} follows by defining the time sharing random variable $Q$ to be uniformly distributed over $[1:n]$, and
\begin{align}\label{eq:powerconst}
E[X^2_Q]=  \frac{1}{n}\sum_{i=1}^{n} E[X^2_i]  =  \frac{1}{n} E\left[\sum_{i=1}^{n}X^2_i\right]  \leq P.
\end{align}

Moreover, letting $I_n:=f_n(Z^n)$, we have for any $\mu >0$ and sufficiently large $n$, 
\begin{align}
n R &=  H(M)\nonumber \\
&=I(M;Y^n,I_n)+H(M|Y^n,I_n)\nonumber \\
&\leq I (X^n;Y^n,I_n)+n\mu \label{eq:intN1<N2}  \\
&=I(X^n;Y^n)+I(X^n;I_n|Y^n)+ n\mu  \nonumber \\
&=I(X^n;Y^n)+H(I_n|Y^n)-H(I_n|X^n)+ n\mu  \label{E:inviewof} \\
&\leq n(I(X_Q;Y_Q)+R_0 -a_n+ \mu ), \nonumber
\end{align}
i.e.,
\begin{align}
 R \leq  I(X_Q;Y_Q)+R_0 -a_n+ \mu ,  \label{E:summarize2}
\end{align}
where $a_n:=\frac{1}{n}H(I_n|X^n)$  satisfies
\begin{align}\label{E:constraint on a}
0 \leq a_n \leq R_0.
\end{align}
Note that in \eqref{E:inviewof} we use the fact that $H(I_n|Y^n,X^n)=H(I_n|X^n)$ due to the Markov chain $I_n-X^n-Y^n$.
\begin{figure}
\centering
\includegraphics[width=0.35\textwidth]{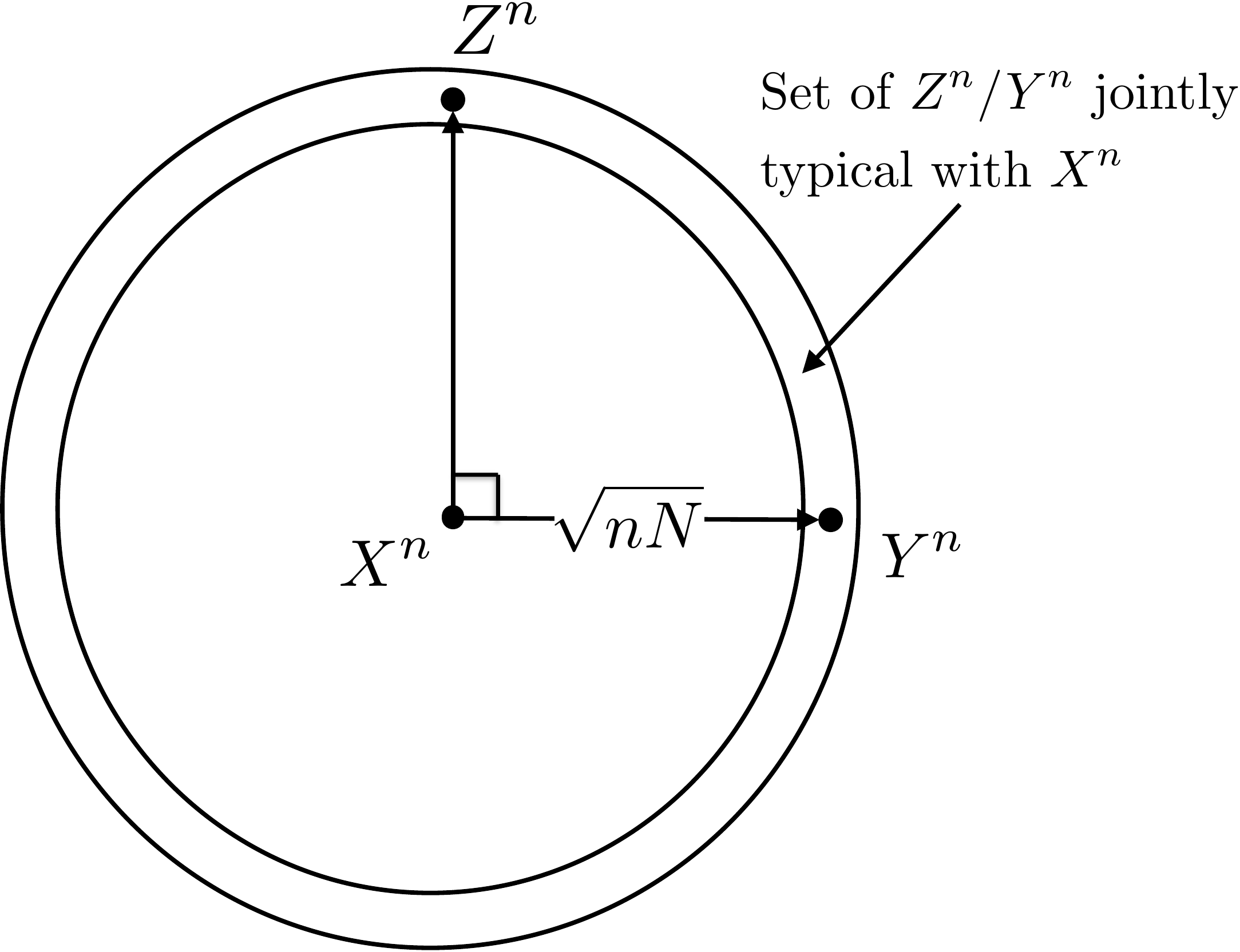}
\caption{Jointly typical set with $X^n$.}
\label{F:Typicalset}
\end{figure}

So far we have made only standard information theoretic arguments and in particular recovered the cut-set bound; note that the fact that $a_n\geq 0$ together with \eqref{E:summarize1}, \eqref{eq:powerconst} and \eqref{E:summarize2} yields the cut-set bound given in Proposition~\ref{P:cutset}. However, instead of simply lower bounding $a_n$ by $0$ in \eqref{E:summarize2}, in the sequel we will prove a third inequality involving $a_n$ that forces $a_n$ to be strictly larger than $0$. Indeed, it is intuitively easy to see that $a_n$ can not be arbitrarily small. Assume $a_n=\frac{1}{n}H(I_n|X^n)\approx 0$. Roughly speaking, this implies that given the transmitted codeword $X^n$, there is no ambiguity about $I_n$, or equivalently all $Z^n$ sequences jointly typical with $X^n$ are mapped to the same $I_n$. See Figure~\ref{F:Typicalset}. However, since $Y^n$ and $Z^n$ are statistically equivalent given $X^n$ (they share the same typical set given $X^n$) this would imply that
$I_n$ can be determined based on $Y^n$ and therefore $H(I_n|Y^n)\approx 0$, which forces the rate to be even smaller than $I(X_Q;Y_Q)$ in view of \dref{E:inviewof}. In general, there is a trade-off between how close the rate can get to the multiple-access bound $I(X_Q;Y_Q)+R_0$ and how much it can exceed the point-to-point capacity $I(X_Q;Y_Q)$ of the $X$-$Y$ link. We capture this trade-off as follows.


Adding and subtracting $H(I_n)$ to the R.H.S. of \eqref{E:inviewof}, we have
\begin{align}
n R
&\leq  I(X^n;Y^n) + I (X^n; I_n) - I(Y^n;I_n) +n\mu. \label{eq:int2}
\end{align}
In Section \ref{SS:Provingineq}, we prove the following key lemma, which enables us to upper bound $I (X^n; I_n) - I(Y^n;I_n)$ in the above inequality.

\begin{lemma}\label{eq:upperboundlemma}
Consider any discrete random vector $X^n\in\mathbb{R}^n$. Let $Z^n=X^n+W_1^n$ and $Y^n=X^n+W_2^n$, where both $W_1^n$ and $W_2^n$ are i.i.d. sequences of Gaussian random variables with zero mean and variance $N$ and they are independent of each other and $X^n$. Also let $I_n=f_n(Z^n)$ be a function of $Z^n$ which takes value on a finite set.  Then, if
$H(I_n|X^n)=na_n$, we have
\begin{equation}\label{eq:upperbound}
 I (X^n; I_n) - I(Y^n;I_n)  \leq n( a_n+\sqrt{2a_n\ln2}\log e  ).
\end{equation}
\end{lemma}

Note that $I_n-Z^n-X^n-Y^n$ in the above lemma form a Markov chain and the result of the lemma can be equivalently regarded as fixing $I (X^n; I_n)= H(I_n) - na_n$ and controlling the second mutual information $I(Y^n;I_n)$. In this sense, there is some similarity in flavor between our result \dref{eq:upperbound} and the strong data processing inequality \cite{AGSDPI}. However, when deriving strong data processing inequalities one is typically interested in upper bounding $I(Y^n;I_n)$ while we are interested in lower bounding it. Moreover, here we assume more specific structure for the Markov chain $I_n-Z^n-X^n-Y^n$.

Note that the random variables $(I_n, Z^n,X^n, Y^n)$ associated with our relay channel trivially satisfy the conditions of Lemma \ref{eq:upperboundlemma}. In particular, $X^n$ in our case is a discrete random vector whose distribution is dictated by the uniform distribution on the set of possible messages and the source codebook, $Y^n$ and $Z^n$ are continuous random vectors and $I_n$ is an integer valued random variable. In light of this, Lemma \ref{eq:upperboundlemma} combined with \eqref{eq:int2} immediately yields that
\begin{align*}
n R&\leq  n(I(X_Q;Y_Q)+a_n+\sqrt{2a_n\ln2}\log e  +\mu) ,
\end{align*}
i.e.,
\begin{align}
R \leq I(X_Q;Y_Q)+a_n+\sqrt{2a_n\ln2}\log e   +\mu.   \label{E:summarize3}\end{align}

Combining \dref{E:summarize1}, \dref{E:summarize2} and \dref{E:summarize3}, we conclude that if a rate $R$ is achievable, then for any $\mu >0$ and sufficiently large $n$,
\begin{numcases}{}
 R   \leq I(X_Q;Y_Q,Z_Q)+\mu  \nonumber \\
R    \leq   I(X_Q;Y_Q)+R_0  -a_n +\mu  \nonumber \\
R    \leq   I(X_Q;Y_Q)+a_n+\sqrt{2a_n\ln2}\log e   +\mu  \nonumber
\end{numcases}
where $E[X_Q^2]\leq P$ and $a_n\in [0,R_0]$. Since $\mu$ can be made arbitrarily small, this proves Proposition \ref{P:newboundsym} and Theorem \ref{T:newboundsym} for  the symmetric case.

\subsection{Proof of Lemma~\ref{eq:upperboundlemma} }\label{SS:Provingineq}


The remaining step then is to prove Lemma~\ref{eq:upperboundlemma}. To prove this lemma we will look at $B$-length i.i.d. sequences of the random vectors $X^n, Y^n, Z^n,$ and $I_n$,  and derive some typicality properties for these sequences which hold with high probability when $B$ is large.

Specifically, consider the following $B$-length i.i.d. sequence
\begin{align}
\{   (X^n(b),Y^n(b),Z^n(b),I_n(b)   )    \}_{b=1}^{B}, \label{E:iidextension}
\end{align}
where for any $b\in [1:B]$, $(X^n(b),Y^n(b),Z^n(b),I_n(b))$ has the same distribution as $(X^n, Y^n, Z^n,I_n)$.  For notational convenience,  in the sequel we write the $B$-length sequence $[X^n(1),X^n(2),\ldots,X^n(B)]$ as $\mathbf X$ and similarly define $\mathbf Y, \mathbf Z$ and $\mathbf I$; note here we have
$\mathbf I=[f_n(Z^n(1)), f_n(Z^n(2)),\ldots,f_n(Z^n(B)) ]=:f(\mathbf Z)$.

We now present a key lemma in  our proof, which gives a lower bound on the conditional probability density $f(\mathbf y| \mathbf i)$ for a set of ``typical'' $(\mathbf y, \mathbf i)$ pairs. The proof of this lemma will be delayed until we  finish proving Lemma~\ref{eq:upperboundlemma}.


\medskip
\begin{lemma}\label{L:Keylemma}
For any $\delta>0$ and sufficiently large $B$, there exists a set $\mathcal I$ of $\mathbf{i}$ such that
 $$\mbox{Pr}(\mathbf I \in \mathcal I) \geq  1-\delta,$$ and for any $\mathbf i \in \mathcal I$, there exists a set $\mathcal{Y}_\mathbf{i}$ of $\mathbf y$ satisfying
 $$\mbox{Pr}(\mathbf Y \in \mathcal{Y}_\mathbf{i}|\mathbf i) \geq  1-\delta,$$ and for any $\mathbf y \in \mathcal{Y}_\mathbf{i}$
\begin{align*}
f(\mathbf{y}|\mathbf{i})\geq 2^{-B(H(X^n|I_n)-H(X^n|Z^n) + \frac{n}{2}\log 2\pi eN +  na_n+n \sqrt{2a_n\ln2}\log e    + n\delta_1)}, \end{align*}
where $\delta_1 \to 0$ as $\delta \to 0$.
\end{lemma}

Equipped with this lemma, it is not difficult to prove Lemma~\ref{eq:upperboundlemma}. For this, first consider $h(\mathbf{Y} |\mathbf{i})$ for any $\mathbf{i} \in \mathcal I$. We have
\begin{align}
h(\mathbf{Y} |\mathbf{i}) &\leq h(\mathbf{Y} |\mathbf{i})+1 - I(\mathbf{Y}; \mathbb{I}(\mathbf{Y} \in \mathcal{Y}_\mathbf{i})|\mathbf{i} ) \label{E:indicator} \\
\nonumber &=  1+h(\mathbf{Y} |\mathbb{I}(\mathbf{Y} \in \mathcal{Y}_\mathbf{i}), \mathbf{i}) \\
&=1+\mbox{Pr}(\mathbf{Y} \in \mathcal{Y}_\mathbf{i}|\mathbf{i} ) h(\mathbf{Y} |\mathbf{i},\mathbf{Y} \in \mathcal{Y}_\mathbf{i}) +\mbox{Pr}(\mathbf{Y} \notin \mathcal{Y}_\mathbf{i}|\mathbf{i} ) h(\mathbf{Y} |\mathbf{i},\mathbf{Y} \notin \mathcal{Y}_\mathbf{i}),\label{E:pluggingH(Y|i)}
\end{align}
where $\mathbb I(A)$ is the indicator function defined as 1 if $A$ holds and 0 otherwise, and  \dref{E:indicator} follows since
$$I(\mathbf{Y}; \mathbb{I}(\mathbf{Y} \in \mathcal{Y}_\mathbf{i})|\mathbf{i} )\leq H(\mathbb{I}(\mathbf{Y} \in \mathcal{Y}_\mathbf{i})|\mathbf{i} )\leq 1.$$

To bound $h(\mathbf{Y} |\mathbf{i},\mathbf{Y} \in \mathcal{Y}_\mathbf{i})$, we have by Lemma \ref{L:Keylemma} that,
\begin{align}
&h(\mathbf{Y} |\mathbf{i},\mathbf{Y} \in \mathcal{Y}_\mathbf{i})\nonumber \\
& =- \int_{\mathbf y \in \mathcal{Y}_\mathbf{i}} f(\mathbf y|\mathbf{i},\mathbf{Y} \in \mathcal{Y}_\mathbf{i} )
\log f(\mathbf y|\mathbf{i},\mathbf{Y} \in \mathcal{Y}_\mathbf{i} ) d\mathbf y \nonumber \\
&\leq  - \int_{\mathbf y \in \mathcal{Y}_\mathbf{i}} f(\mathbf y|\mathbf{i},\mathbf{Y} \in \mathcal{Y}_\mathbf{i} )
\log f(\mathbf y|\mathbf{i} ) d\mathbf y \nonumber \\
&\leq B\left(H(X^n|I_n)-H(X^n|Z^n) + \frac{n}{2}\log 2\pi eN +  na_n+n \sqrt{2a_n\ln2}\log e    + n\delta_1\right) \cdot \int_{\mathbf y \in \mathcal{Y}_\mathbf{i}} f(\mathbf y|\mathbf{i},\mathbf{Y} \in \mathcal{Y}_\mathbf{i} )
  d\mathbf y \nonumber \\
&= B\left(H(X^n|I_n)-H(X^n|Z^n) + \frac{n}{2}\log 2\pi eN +  na_n+n \sqrt{2a_n\ln2}\log e    + n\delta_1\right). \label{E:pluggedH(Y|i)1}
\end{align}

Now consider $E [   \|\mathbf{Y}\|^2   |\mathbf{i}] $ for any $\mathbf i$. We have
\begin{align*}
E [   \|\mathbf{Y}\|^2   |\mathbf{i}] &=  E [   \|\mathbf{X}\|^2   |\mathbf{i}] +  E [   \|\mathbf{W_2}\|^2   |\mathbf{i}]\leq nB(P+N),
\end{align*}
where the equality follows from the independence between $\mathbf{X}$ and $\mathbf{W_2}$ even conditioned on $\mathbf{i}$.
Therefore,    $$E [   \|\mathbf{Y}\|^2   |\mathbf{i},\mathbf{Y} \notin \mathcal{Y}_\mathbf{i}] \leq \frac{E [   \|\mathbf{Y}\|^2   |\mathbf{i}] }{   \mbox{Pr}(\mathbf{Y} \notin \mathcal{Y}_\mathbf{i}|\mathbf{i} )  }\leq \frac{nB(P+N)}{   \mbox{Pr}(\mathbf{Y} \notin \mathcal{Y}_\mathbf{i}|\mathbf{i} )  }  ,$$
and
\begin{align}
&\ ~\ ~\mbox{Pr}(\mathbf{Y} \notin \mathcal{Y}_\mathbf{i}|\mathbf{i} ) h(\mathbf{Y} |\mathbf{i},\mathbf{Y} \notin \mathcal{Y}_\mathbf{i})\nonumber \\
& \leq\frac{nB}{2} \mbox{Pr}(\mathbf{Y} \notin \mathcal{Y}_\mathbf{i}|\mathbf{i} )  \log 2\pi e \frac{P+N}{   \mbox{Pr}(\mathbf{Y} \notin \mathcal{Y}_\mathbf{i}|\mathbf{i} )} \nonumber \\
& \leq  nB \delta_2 ,\label{E:pluggedH(Y|i)2}
\end{align}
for some $\delta_2 \to 0$ as $\delta \to 0$.

Plugging \dref{E:pluggedH(Y|i)1}  and \dref{E:pluggedH(Y|i)2}   into \dref{E:pluggingH(Y|i)}, we have for any $\mathbf{i} \in \mathcal I$,
\begin{align}
h(\mathbf{Y} |\mathbf{i}) &\leq 1+ \mbox{Pr}(\mathbf{Y} \in \mathcal{Y}_\mathbf{i}|\mathbf{i} ) B\left(H(X^n|I_n)-H(X^n|Z^n) + \frac{n}{2}\log 2\pi eN +  na_n+n \sqrt{2a_n\ln2}\log e + n\delta_1\right) \nonumber \\
&~~~~~+ nB\delta_2\nonumber \\
\nonumber  &=  B\left(H(X^n|I_n)-H(X^n|Z^n) + \frac{n}{2}\log 2\pi eN +  na_n+n \sqrt{2a_n\ln2}\log e + n\delta_3\right)
\end{align}
where $\delta_3 \to 0$ as $\delta \to 0$ and $B\to \infty$. Therefore, for sufficiently large $B$,
\begin{align}
h(\mathbf{Y}|\mathbf{I})&=\sum_{i} p(\mathbf{i})h(\mathbf{Y} |\mathbf{i})\nonumber \\
&= \sum_{\mathbf{i}\in \mathcal I} p(\mathbf{i})h(\mathbf{Y} |\mathbf{i}) + \sum_{\mathbf{i} \not\in \mathcal I} p(\mathbf{i})h(\mathbf{Y}|\mathbf{i})\nonumber \\
&\leq \sum_{\mathbf{i}\in \mathcal I} p(\mathbf{i})B\left(H(X^n|I_n)-H(X^n|Z^n) + \frac{n}{2}\log 2\pi eN +  na_n+n \sqrt{2a_n\ln2}\log e + n\delta_3\right)\nonumber\\
&~~~+ \sum_{\mathbf{i} \not\in \mathcal I} p(\mathbf{i})  \frac{nB}{2}\log 2\pi e (P+N)       \nonumber \\
&= B\left(H(X^n|I_n)-H(X^n|Z^n) + \frac{n}{2}\log 2\pi eN +  na_n+n \sqrt{2a_n\ln2}\log e + n\delta_4\right) \nonumber
\end{align}
where $\delta_4 \to 0$ as $\delta \to 0$ and $B\to \infty$.
Observing that
\begin{align}
h(\mathbf{Y}|\mathbf{I}) = \sum_{b=1}^{B}h( Y^n(b) |  I_n(b)) =Bh(Y^n|I_n)\nonumber
\end{align}
and taking $B\to \infty$, we obtain
\begin{align}
h(Y^n|I_n) \leq
 H(X^n|I_n)-H(X^n|Z^n) + \frac{n}{2}\log 2\pi eN +  na_n+n \sqrt{2a_n\ln2}\log e. \label{E:relationfinal}
\end{align}

Finally, using the relation \dref{E:relationfinal}, we have
\begin{align*}
I (X^n; I_n) - I(Y^n;I_n) & = H(X^n)-H(X^n|I_n) - h(Y^n)  +  h(Y^n|I_n) \nonumber \\
& =  H(X^n)-H(X^n|I_n) -  h(Y^n) + H(X^n|I_n) \nonumber \\
&~~~ -H(X^n|Z^n) + \frac{n}{2}\log 2\pi eN +  na_n+n \sqrt{2a_n\ln2}\log e \nonumber \\
& =  I(X^n;Z^n) - [ h(Y^n) - \frac{n}{2}\log 2\pi eN] +  na_n+n \sqrt{2a_n\ln2}\log e \nonumber \\
& =  I(X^n;Z^n) - I(X^n;Y^n) +  na_n+n \sqrt{2a_n\ln2}\log e \nonumber \\
& =   na_n+n \sqrt{2a_n\ln2}\log e
\end{align*}
where the last step follows from the symmetry of the channel, i.e. $I(X^n;Z^n) = I(X^n;Y^n) $. This finishes the proof of Lemma~\ref{eq:upperboundlemma}.

\subsection{Proof Outline for Lemma~\ref{L:Keylemma}}\label{SS:outlinesym}

We now provide a proof sketch for Lemma~\ref{L:Keylemma} that summarizes the main ideas. The formal proof is rather technical and will be given in the next subsection.

By the law of large numbers, if $H(I_n|X^n)=na_n$, then given a typical $(\mathbf x, \mathbf i)$ pair, it can be shown that
$$\mbox{Pr}(\mathbf{Z} \in \mathcal Z_{(\mathbf x, \mathbf i)}|\mathbf x)\doteq2^{-B\,I(Z^n;I_n|X^n)}= 2^{-nBa_n},$$
where $\mathcal Z_{(\mathbf x, \mathbf i)}$ can be roughly viewed as the set of $\mathbf z$ that are jointly typical with $(\mathbf x, \mathbf i)$.

Now we apply the following lemma, whose proof relies on a Gaussian measure concentration result and is included in Appendix \ref{A:ProoftoTalagrand}.
\begin{lemma}\label{L:Talagrand}
Let $U_1, U_2, \ldots, U_n$ be $n$ i.i.d. Gaussian random variables with $U_i \sim \mathcal{N}(0,N),\forall i \in \{1,2,\ldots,n\}$. Then, for any $A\subseteq \mathbb R^{n}$ with $\mbox{Pr}(U^n \in A)\geq 2^{-na_n}$,
\begin{align*}
\mbox{Pr}(U^n \in \Gamma_{\sqrt{n}( \sqrt{2Na_n\ln2} +r ) } (A) )\geq 1-2^{-\frac{nr^2}{2N}},\forall r>0,
\end{align*}
where
\begin{align*}
&\Gamma_{\sqrt{n}( \sqrt{2Na_n\ln2} +r ) } (A):=\{\underline{\mathbf{\omega}}\in \mathbb R^{n}: \exists \  \underline{\mathbf{\omega}}'\in A \text{~s.t.~} d(\underline{\mathbf{\omega}},\underline{\mathbf{\omega}}') \leq \sqrt{n} (\sqrt{2Na_n\ln2} +r ) \},
\end{align*}
with $d(\underline{\mathbf{\omega}},\underline{\mathbf{\omega}}') :=\|\underline{\mathbf{\omega}}-\underline{\mathbf{\omega}}' \|$ denoting the Euclidean distance between $\underline{\mathbf{\omega}}$ and $\underline{\mathbf{\omega}}'$.
\end{lemma}

With Lemma  \ref{L:Talagrand}, it can be shown that if one blows up $\mathcal Z_{(\mathbf x, \mathbf i)}$ with a radius $\sqrt{nB} \sqrt{2Na_n\ln2} $, the resultant set, denoted by $\Gamma_{\sqrt{nB} \sqrt{2Na_n\ln2} }(\mathcal Z_{(\mathbf x, \mathbf i)})$, has probability nearly 1, i.e.,
\begin{align}\label{E:blowingup}
\mbox{Pr}(\mathbf Z \in \Gamma_{\sqrt{nB} \sqrt{2Na_n\ln2}}(\mathcal Z_{(\mathbf x, \mathbf i)})|\mathbf x)\approx 1.
\end{align}
Due to the symmetry of the channel, \dref{E:blowingup} still holds with $\mathbf Z$ replaced by $\mathbf Y$. 

Now given a typical $(\mathbf x, \mathbf i)$ pair, we will lower bound the conditional density $f(\mathbf y|\mathbf i)$ for all $\mathbf y\in\Gamma_{\sqrt{nB} \sqrt{2Na_n\ln2}}(\mathcal Z_{(\mathbf x, \mathbf i)})$. In particular, given such $\mathbf y$, there exists some $\mathbf{z}\in \mathcal Z_{(\mathbf x, \mathbf i)}$ such that $ d(\mathbf y ,  \mathbf{z} )\leq \sqrt{nB} \sqrt{2Na_n\ln2}$. Consider the set of all $\mathbf x$'s that are jointly typical with this $\mathbf{z}$. 
It can be shown that the $\mathbf x$'s that are jointly typical with a given $\mathbf{z}\in \mathcal Z_{(\mathbf x, \mathbf i)}$ are such that
$$
d(\mathbf x, \mathbf z)\leq \sqrt{nBN},
$$
and
$$p(\mathbf x|\mathbf i)\doteq 2^{-BH(X^n|I_n)}.$$
 Therefore for each $\mathbf x$ in this set
\begin{align*}
d(\mathbf x, \mathbf y)&\leq d(\mathbf x, \mathbf z)+d(\mathbf z, \mathbf y)\\
&\leq \sqrt{nB}(\sqrt{N}+\sqrt{2Na_n\ln 2}),
\end{align*}
which leads to the following lower bound on $f(\mathbf y|\mathbf x)$,
$$f(\mathbf y|\mathbf x) \stackrel {.}{\geq}  2^{ -nB \left( \frac{1}{2}\log 2\pi e N +a_n+\sqrt{2a_n\ln2}\log e       \right) },$$
by using the fact that $\mathbf y$ is Gaussian given $\mathbf x$. The set of such $\mathbf x$'s can be shown to have cardinality approximately given by $2^{BH(X^n|Z^n)}$. Combining this with the above, we have
\begin{align*}
f(\mathbf{y}|\mathbf{i})&=\sum_{\mathbf x} f(\mathbf{y}|\mathbf{x})p(\mathbf{x}|\mathbf{i})\nonumber \\
&\stackrel {.}{\geq}  2^{BH(X^n|Z^n)} 2^{-BH(X^n|I_n)}  2^{ -nB \left( \frac{1}{2}\log 2\pi e N +a_n+\sqrt{2a_n\ln2}\log e      \right) }.
\end{align*}
Using the fact  that $(\mathbf x, \mathbf i)$ are jointly typical with high probability and given a typical $(\mathbf x, \mathbf i)$ the above lower bound holds for all $\mathbf{y}$ with high probability concludes the proof sketch of Lemma \ref{L:Keylemma}. A rigorous proof is given in the sequel.

\subsection{Formal Proof of Lemma \ref{L:Keylemma}}\label{SS:ProoftoKeylemma}

\subsubsection{Definitions of High Probability Sets}
By considering the $B$-length i.i.d. extensions of the $n$-letter random variables involved, the law of large numbers allows us to concentrate on a series of ``high probability'' sets defined in the following.\footnote{The high probability sets defined here are analogous to  strongly typical sets \cite{ElGamalKim} that are widely used in information theory. However, in the Gaussian case the notion of strong typicality doesn't apply and thus we need to develop our own customized high probability sets. In the discrete memoryless case \cite{WuOzgurXie_TIT}, one can simply resort to strong typicality.}

\vspace{2mm}
\noindent \underline{Definition of $\tilde S(X^n,Z^n)$}
\vspace{2mm}
\begin{lemma}\label{L:unchange1}
Assume $H(I_n| X^{n})={n}a_n, H(X^{n}|I_{n})={n}b_n, H(X^{n}|Z^{n})={n}c_n$ for the $n$-channel use code.
Given any $\epsilon >0$ and sufficiently large $B$, we have
\begin{align*}
\mbox{Pr}((\mathbf{X},\mathbf{Z})\in \tilde S(X^n,Z^n))\geq 1-\epsilon
\end{align*}
where
\begin{align*}
 \tilde S(X^n,Z^n):=\Big \{(\mathbf{x},\mathbf{z} ): & ~d(\mathbf x, \mathbf z) \in [\sqrt{nB}(\sqrt{N}-\epsilon),  \sqrt{nB}(\sqrt{N}+\epsilon)    ]\\
 &~2^{-nB(a_n+\epsilon)}\leq p( f(\mathbf{z})|\mathbf{x})\leq 2^{-nB(a_n-\epsilon)} \\
 &~2^{-nB(b_n+\epsilon)}\leq p(\mathbf{x}|f(\mathbf{z}))\leq 2^{-nB(b_n-\epsilon)} \\
  &~2^{-nB(c_n+\epsilon)}\leq p(\mathbf{x}|\mathbf{z})\leq 2^{-nB(c_n-\epsilon)} \Big \}
\end{align*}
\end{lemma}
The lemma is a simple consequence of the law of large numbers.

\vspace{2mm}
\noindent \underline{Definition of $S(X^n,Z^n)$}
\vspace{2mm}

To define $S(X^n,Z^n)$,  we first consider the following lemma, which has been proved in \cite{Zhang}.
\begin{lemma}\label{L:zhanglemma}
Let $A \subseteq C\times D$. For $x\in C$, use $A|_x$ to denote the set $$A|_x=\{y\in D: (x,y)\in A \}.$$ If $\mbox{Pr}(A)\geq 1-\epsilon$,  then
$\mbox{Pr}(B)\geq 1-\sqrt{\epsilon}$,
where $$B:=\{x\in C: \mbox{Pr}(A|_x|x)\geq 1-\sqrt{\epsilon} \}.$$
\end{lemma}

Now, define
\begin{align*}
S(X^n,Z^n)=\{ (\mathbf{x},\mathbf{z})\in \tilde S(X^n,Z^n): \mbox{Pr}(\tilde S(X^n,Z^n)|_\mathbf{z}|\mathbf{z})\geq 1-\sqrt{\epsilon}  \}.
\end{align*}
Clearly $S(X^n,Z^n)$ is a subset of $\tilde S(X^n,Z^n)$. The following lemma says that it is also a high probability set.
\begin{lemma} \label{L: properties_S(x,z)}
$\mbox{Pr}(S(X^n,Z^n))\geq 1- 2\sqrt{\epsilon}$ for $B$ sufficiently large.
\end{lemma}

\begin{proof}
Consider $B$ sufficiently large. Due to Lemma \ref{L:zhanglemma} and the fact that $\mbox{Pr}(\tilde S(X^n,Z^n) )\geq 1-\epsilon$, we have
$$\mbox{Pr}\{ (\mathbf{x},\mathbf{z}): \mbox{Pr}(\tilde S(X^n,Z^n)|_\mathbf{z}|\mathbf{z})\geq 1-\sqrt{\epsilon}  \}\geq 1-\sqrt{\epsilon}.$$
Then by the definition of $S(X^n,Z^n)$,
\begin{align*}
&\mbox{Pr}(S^c(X^n,Z^n))\\
\leq\ & \mbox{Pr}(\tilde S^c(X^n,Z^n))+
\mbox{Pr}\{ (\mathbf{x},\mathbf{z}): \mbox{Pr}(\tilde S(X^n,Z^n)|_\mathbf{z}|\mathbf{z})< 1-\sqrt{\epsilon}  \}\\
\leq \ & \epsilon+\sqrt{\epsilon}\\
\leq \ & 2\sqrt{\epsilon},
\end{align*}
and thus $\mbox{Pr}(S(X^n,Z^n))\geq 1- 2\sqrt{\epsilon}$.
%
%
\end{proof}

\vspace{2mm}
\noindent \underline{Definitions of $\mathcal{Z}_{(\mathbf{x},\mathbf{i})}$ and $S(X^n,I_n)$}
\vspace{2mm}

Define $$\mathcal{Z}_{(\mathbf{x},\mathbf{i})}=\{ \mathbf{z}:  f(\mathbf{z})=\mathbf{i}, (\mathbf{x},\mathbf{z})\in S(X^n,Z^n)   \}$$
and $$S(X^n,I_n)=\{ (\mathbf{x},\mathbf{i}): \mbox{Pr}(\mathcal{Z}_{(\mathbf{x},\mathbf{i})}|\mathbf{x},\mathbf{i})\geq 1-\sqrt[4]{\epsilon} \}.$$

\begin{lemma}\label{L:prob_s(x,i)}
$\mbox{Pr}(S(X^n,I_n))\geq 1-2\sqrt[4]{\epsilon}$ for $B$ sufficiently large.
\end{lemma}
\begin{proof}
For $B$ sufficiently large, consider $\mbox{Pr}( \mathbf{Z} \notin \mathcal{Z}_{(\mathbf{X},\mathbf{I})}     )$.
We have
\begin{align*}
\mbox{Pr}( \mathbf{Z} \notin \mathcal{Z}_{(\mathbf{X},\mathbf{I})}   )=\mbox{Pr}( f(\mathbf{Z})=\mathbf{I},  (\mathbf{X},\mathbf{Z})\notin  S(X^n,Z^n)    )\leq 2 \sqrt{\epsilon}.
\end{align*}

On the other hand,
\begin{align*}
\mbox{Pr}( \mathbf{Z} \notin \mathcal{Z}_{(\mathbf{X},\mathbf{I})}    )
&=\sum_{(\mathbf{x},\mathbf{i})\in S(X^n,I_n)}\mbox{Pr}(\mathbf{Z} \notin \mathcal{Z}_{(\mathbf{x},\mathbf{i})} | \mathbf{x},\mathbf{i}   ) p(\mathbf{x},\mathbf{i})\\
&~~+\sum_{(\mathbf{x},\mathbf{i})\notin S(X^n,I_n)}\mbox{Pr}(\mathbf{Z} \notin \mathcal{Z}_{(\mathbf{x},\mathbf{i})} | \mathbf{x},\mathbf{i}   ) p(\mathbf{x},\mathbf{i})\\
&\geq \sqrt[4]{\epsilon}\cdot \mbox{Pr}( S^c (X^n,I_n)  ).
\end{align*}

Therefore, $\mbox{Pr}( S^c (X^n,I_n)  ) \leq 2\sqrt{\epsilon}/\sqrt[4]{\epsilon}= 2\sqrt[4]{\epsilon}$, and $\mbox{Pr}(S(X^n,I_n))\geq 1-2\sqrt[4]{\epsilon}$.
\end{proof}

\begin{lemma}\label{L: properties_S(x,i)}
For any $(\mathbf{x},\mathbf{i}) \in S(X^n,I_n)$, we have
$$2^{-nB(a_n+\epsilon)}\leq p( \mathbf{i} |\mathbf{x})\leq 2^{-nB(a_n-\epsilon)},$$
and for sufficiently large $B$,
$$\mbox{Pr}(\mathcal{Z}_{(\mathbf{x},\mathbf{i})}|\mathbf{x})\geq 2^{-nB( a_n+2\epsilon )}.$$

\end{lemma}
\begin{proof}
Consider any $(\mathbf{x},\mathbf{i}) \in S(X^n,I_n)$. From the definition of $S(X^n,I_n)$,
$\mbox{Pr}(\mathcal{Z}_{(\mathbf{x},\mathbf{i})}|\mathbf{x},\mathbf{i})\geq 1-\sqrt[4]{\epsilon}$. Therefore,
$\mathcal{Z}_{(\mathbf{x},\mathbf{i})}$ must be nonempty, i.e., there exists at least one $\mathbf{z}\in \mathcal{Z}_{(\mathbf{x},\mathbf{i})}$.

Consider any $\mathbf{z}\in \mathcal{Z}_{(\mathbf{x},\mathbf{i})}$. By the definition of $\mathcal{Z}_{(\mathbf{x},\mathbf{i})}$,
 we have $f(\mathbf{z})=\mathbf{i}$ and $(\mathbf{x},\mathbf{z})\in S(X^n,Z^n)\subseteq \tilde S(X^n,Z^n)$.  Then, it follows from the definition of $\tilde S(X^n,Z^n)$ that
$$2^{-nB(a_n+\epsilon)}\leq p( f(\mathbf{z}) |\mathbf{x})\leq 2^{-nB(a_n-\epsilon)},$$ i.e.,  $$2^{-nB(a_n+\epsilon)}\leq p( \mathbf{i}|\mathbf{x})\leq 2^{-nB(a_n-\epsilon)}.$$
Furthermore,
\begin{align*}
 \mbox{Pr}(\mathbf{Z}\in \mathcal{Z}_{(\mathbf{x},\mathbf{i})}|\mathbf{x})
=\ &\frac{\mbox{Pr}(f(\mathbf{Z})=\mathbf{i}|\mathbf{x})\mbox{Pr}(\mathbf{Z}\in \mathcal{Z}_{(\mathbf{x},\mathbf{i})}|\mathbf{x},f(\mathbf{Z})=\mathbf{i})}
             {\mbox{Pr}(f(\mathbf{Z})=\mathbf{i}|\mathbf{Z}\in \mathcal{Z}_{(\mathbf{x},\mathbf{i})},\mathbf{x})}\\
=\ &p( \mathbf{i}|\mathbf{x}) \mbox{Pr}(\mathcal{Z}_{(\mathbf{x},\mathbf{i})}|\mathbf{x}, \mathbf{i})\\
\geq \ & 2^{-nB(a_n+\epsilon)} (1-\sqrt[4]{\epsilon})\\
\geq\ &2^{-nB(a_n+2\epsilon)}
\end{align*}
for sufficiently large $B$. This finishes the proof of the lemma.
\end{proof}

\subsubsection{Blowing Up $\mathcal{Z}_{(\mathbf{x},\mathbf{i})}$}

\begin{lemma}\label{L:Blown-Up}
For any $(\mathbf{x},\mathbf{i}) \in S(X^n,I_n)$, consider the following blown-up set of $\mathcal{Z}_{(\mathbf{x},\mathbf{i})}$:
\begin{align*}
&\Gamma_{\sqrt{nB} (\sqrt{2 N a_n}+3\sqrt{N\epsilon}) }(\mathcal{Z}_{(\mathbf{x},\mathbf{i})}) =\{\underline{\mathbf{\omega}}\in \mathbb R^{nB}: \exists \  \underline{\mathbf{\omega}}'\in \mathcal{Z}_{(\mathbf{x},\mathbf{i})} \text{~~s.t.~~} d(\underline{\mathbf{\omega}},\underline{\mathbf{\omega}}')\leq \sqrt{nB} (\sqrt{2 N a_n}+3\sqrt{N\epsilon}) \}.
\end{align*}
We have
\begin{enumerate}
    \item $\mbox{Pr}(\mathbf{Y}\in \Gamma_{\sqrt{nB}  (\sqrt{2Na_n}+3\sqrt{N\epsilon}) }(\mathcal{Z}_{(\mathbf{x},\mathbf{i})}) |\mathbf{x}) \geq 1-\epsilon$ for sufficiently large $B$;
    \item For any $\mathbf{y}\in \Gamma_{\sqrt{nB}  (\sqrt{2Na_n}+3\sqrt{N\epsilon})  }(\mathcal{Z}_{(\mathbf{x},\mathbf{i})}) $,
$$f(\mathbf{y}|\mathbf{i})\geq 2^{-nB(b_n-c_n + \frac{1}{2}\log 2\pi eN +  (a_n+\sqrt{2a_n})\log e   + \epsilon')} $$
where $\epsilon' \to 0$ as $\epsilon \to 0$ and $B\to \infty$.
\end{enumerate}
\end{lemma}

\begin{proof} From Lemma \ref{L: properties_S(x,i)}, for any $(\mathbf{x},\mathbf{i}) \in S(X^n,I_n)$ and sufficiently large $B$,
$$\mbox{Pr}(\mathbf{Z} \in \mathcal{Z}_{(\mathbf{x},\mathbf{i})}|\mathbf{x})\geq 2^{-nB( a_n+2\epsilon )},$$
i.e.,
\begin{align*}
\mbox{Pr}(\mathbf x +\mathbf W_1    \in \mathcal{Z}_{(\mathbf{x},\mathbf{i})}|\mathbf{x})& = \mbox{Pr}( \mathbf W_1    \in    \{ \underline{\omega} -\mathbf x: \underline{\omega} \in  \mathcal{Z}_{(\mathbf{x},\mathbf{i})} \} ) \\
&   \geq 2^{-nB( a_n+2\epsilon )}.
\end{align*}

Therefore, we have
\begin{align}
&\mbox{Pr}(\mathbf{Y}\in \Gamma_{\sqrt{nB}  (\sqrt{2Na_n\ln2}+3\sqrt{N\epsilon}) }(\mathcal{Z}_{(\mathbf{x},\mathbf{i})}) |\mathbf{x}) \nonumber \\
=\ & \mbox{Pr}(\mathbf x +\mathbf W_2    \in   \Gamma_{\sqrt{nB}  (\sqrt{2Na_n\ln2}+3\sqrt{N\epsilon}) }(\mathcal{Z}_{(\mathbf{x},\mathbf{i})}) |\mathbf{x})\nonumber \\
=\ &\mbox{Pr}( \mathbf W_2    \in    \{ \underline{\omega} -\mathbf x: \underline{\omega} \in \Gamma_{\sqrt{nB} (\sqrt{2Na_n\ln2}+3\sqrt{N\epsilon}) }(\mathcal{Z}_{(\mathbf{x},\mathbf{i})}) \}   ) \nonumber \\
=\ &\mbox{Pr}( \mathbf W_1    \in    \{ \underline{\omega} -\mathbf x: \underline{\omega} \in \Gamma_{\sqrt{nB} (\sqrt{2Na_n\ln2}+3\sqrt{N\epsilon}) }(\mathcal{Z}_{(\mathbf{x},\mathbf{i})})) \}   ) \nonumber \\
=\ &\mbox{Pr}( \mathbf W_1    \in   \Gamma_{\sqrt{nB} (\sqrt{2Na_n\ln2}+3\sqrt{N\epsilon}) }(  \{ \underline{\omega} -\mathbf x: \underline{\omega} \in \mathcal{Z}_{(\mathbf{x},\mathbf{i})}\} ) ) \nonumber \\
\geq \ &\mbox{Pr}( \mathbf W_1    \in   \Gamma_{\sqrt{nB} (\sqrt{2Na_n\ln2+ 4N\epsilon\ln2}  +\sqrt{N\epsilon}      ) }(  \{ \underline{\omega} -\mathbf x: \underline{\omega} \in \mathcal{Z}_{(\mathbf{x},\mathbf{i})} \} ) ) \nonumber \\
\geq \ & 1-2^{-\frac{nB\epsilon}{2}} \label{E:followfromtalagrand}\\
\geq \ &1-\epsilon \nonumber
\end{align}
for sufficiently large $B$, where \dref{E:followfromtalagrand} follows from Lemma \ref{L:Talagrand}.

To prove Part 2), consider any $\mathbf{y}\in \Gamma_{\sqrt{nB}  (\sqrt{2Na_n\ln2}+3\sqrt{N\epsilon}) }(\mathcal{Z}_{(\mathbf{x},\mathbf{i})})$. We can find one
$\mathbf{z}\in  \mathcal{Z}_{(\mathbf{x},\mathbf{i})}$ such that $d (\mathbf{y},\mathbf{z})\leq \sqrt{nB}  (\sqrt{2Na_n\ln2}+3\sqrt{N\epsilon})$, and for this
$\mathbf{z}$, we have from the definition of $\mathcal{Z}_{(\mathbf{x},\mathbf{i})}$ that: i) $f(\mathbf{z})=\mathbf i$ and ii) $\mbox{Pr}(\tilde S(X^n,Z^n)|_\mathbf{z}|\mathbf{z})\geq 1-\sqrt{\epsilon}$, where
\begin{align*}
\tilde S(X^n,Z^n)|_\mathbf{z}=\Big\{ \mathbf{x}: & ~d(\mathbf x, \mathbf z) \in [\sqrt{nB}(\sqrt{N}-\epsilon),  \sqrt{nB}(\sqrt{N}+\epsilon)    ]\\
 &~2^{-nB(a_n+\epsilon)}\leq p( f(\mathbf{z})|\mathbf{x})\leq 2^{-nB(a_n-\epsilon)}\\
 &~2^{-nB(b_n+\epsilon)}\leq p(\mathbf{x}|f(\mathbf{z}))\leq 2^{-nB(b_n-\epsilon)}  \\
  &~2^{-nB(c_n+\epsilon)}\leq p(\mathbf{x}|\mathbf{z})\leq 2^{-nB(c_n-\epsilon)} \Big\}.
\end{align*}
The size of $\tilde S(X^n,Z^n)|_\mathbf{z}$ can be lower bounded by considering the following
\begin{align*}
1-\sqrt{\epsilon} &\leq \mbox{Pr}(\tilde S(X^n,Z^n)|_\mathbf{z}|\mathbf{z})\\
&= \sum_{\mathbf x \in \tilde S(X^n,Z^n)|_\mathbf{z}} p(\mathbf x|\mathbf z)\\
&\leq 2^{-nB(c_n-\epsilon)} \big|\tilde S(X^n,Z^n)|_\mathbf{z}\big|,
\end{align*}
i.e.,
\begin{align*}\big|\tilde S(X^n,Z^n)|_\mathbf{z}\big| \geq (1-\sqrt{\epsilon} )2^{nB(c_n-\epsilon)}.
\end{align*}

Then,
\begin{align}
f(\mathbf{y}|\mathbf{i})&=\sum_{\mathbf x} f(\mathbf{y}|\mathbf{x})p(\mathbf{x}|\mathbf{i})\nonumber \\
&\geq \sum_{\mathbf x \in \tilde S(X^n,Z^n)|_\mathbf{z}} f(\mathbf{y}|\mathbf{x})p(\mathbf{x}|\mathbf{i})\nonumber \\
&\geq 2^{-nB(b_n+\epsilon)} \sum_{\mathbf x \in \tilde S(X^n,Z^n)|_\mathbf{z}} f(\mathbf{y}|\mathbf{x})\nonumber \\
&\geq 2^{-nB(b_n+\epsilon)} \big|\tilde S(X^n,Z^n)|_\mathbf{z}\big| \min_{\mathbf x \in \tilde S(X^n,Z^n)|_\mathbf{z}} f(\mathbf{y}|\mathbf{x}) \nonumber \\
&\geq  (1-\sqrt{\epsilon} )2^{-nB(b_n+\epsilon)}2^{nB(c_n-\epsilon)} \min_{\mathbf x \in \tilde S(X^n,Z^n)|_\mathbf{z}} f(\mathbf{y}|\mathbf{x}). \label{E:probtobecont}
\end{align}
For any $\mathbf x \in \tilde S(X^n,Z^n)|_\mathbf{z}$, we have
\begin{align*}
d(\mathbf x, \mathbf y)&\leq d(\mathbf x, \mathbf z)+d(\mathbf z, \mathbf y)\\
&\leq \sqrt{nB}(\sqrt{N}+\sqrt{2Na_n\ln2}+\epsilon+3\sqrt{N\epsilon})\\
&=: \sqrt{nB}(\sqrt{N}+\sqrt{2Na_n\ln2}+\epsilon_1)
\end{align*}
and thus,
\begin{align*}
f(\mathbf y| \mathbf x)&=  \frac{1}{ (2\pi N)^{ \frac{nB}{2} } } e^{ -\frac{\| \mathbf y- \mathbf x  \|^2     }{2N}  } \\
&\geq    2^{ -\frac{     nB (\sqrt{N}+\sqrt{2Na_n\ln2}+\epsilon_1)    ^2     }{2N}  \log e  -\frac{nB}{2}\log 2\pi N   } \\
&=    2^{ -nB \left(   \frac{      (\sqrt{N}+\sqrt{2Na_n\ln2}+\epsilon_1)    ^2     }{2N}  \log e + \frac{1}{2}\log 2\pi N  \right) } \\
&=:    2^{ -nB \left( \frac{1}{2}\log 2\pi e N +a_n+\sqrt{2a_n\ln2}\log e   +\epsilon_2  \right) }
\end{align*}
where $\epsilon_1, \epsilon_2 \to 0$ as $\epsilon \to 0$. Plugging this into \dref{E:probtobecont} yields that
\begin{align*}
f(\mathbf y|\mathbf i)&\geq(1-\sqrt{\epsilon} )2^{-nB(b_n+\epsilon)}2^{nB(c_n-\epsilon)}\\
&~~~~\times 2^{ -nB \left( \frac{1}{2}\log 2\pi e N +a_n+\sqrt{2a_n\ln2}\log e    +\epsilon_2  \right) }  \\
&\geq 2^{-nB(b_n-c_n + \frac{1}{2}\log 2\pi e N+a_n+\sqrt{2a_n\ln2}\log e  + \epsilon_3)}
\end{align*}
for some $\epsilon_3 \to 0$ as $\epsilon \to 0$.
\end{proof}

\subsubsection{Constructions of $\mathcal I$ and $\mathcal{Y}_\mathbf{i}$}
Let $\mathcal I=\{\mathbf{i}: \mbox{Pr}(S(X^n,I_n)|_\mathbf{i}|\mathbf{i}) \geq 1-2\sqrt[8]{\epsilon} \}$.  For sufficiently large $B$, $\mbox{Pr}(S(X^n,I_n))\geq 1-2\sqrt[4]{\epsilon}$
from Lemma \ref{L:prob_s(x,i)}, and thus by Lemma \ref{L:zhanglemma} again,
\begin{align*}
\mbox{Pr}(\mathcal I)&\geq  \mbox{Pr}\left\{\mathbf{i}: \mbox{Pr}(S(X^n,I_n)|_\mathbf{i}|\mathbf{i}) \geq 1-   \sqrt{ 2\sqrt[4]{\epsilon}} \right\}\\
&\geq 1-   \sqrt{ 2\sqrt[4]{\epsilon}} \\
&\geq 1-2\sqrt[8]{\epsilon}.
\end{align*}
\begin{lemma}\label{L:AppendixY_i}
For any $\mathbf{i} \in \mathcal I$, let
\begin{align*}
 \mathcal{Y}_\mathbf{i}&:=\bigcup_{\mathbf{x} \in S(X^n,I_n)|_\mathbf{i}} \Gamma_{\sqrt{nB} (\sqrt{2Na_n\ln2}+3\sqrt{N\epsilon})  }(\mathcal{Z}_{(\mathbf{x},\mathbf{i})}).
 \end{align*}
Then for sufficiently large $B$,
$$\mbox{Pr}(\mathbf{Y} \in \mathcal{Y}_\mathbf{i} |\mathbf{i})\geq 1-3\sqrt[8]{\epsilon},$$
and for each $\mathbf y \in \mathcal{Y}_\mathbf{i}$,
\begin{align*}
f(\mathbf{y}|\mathbf{i})\geq 2^{-nB(b_n-c_n + \frac{1}{2}\log 2\pi eN +  a_n+\sqrt{2a_n\ln2}\log e    + \epsilon_3)}. \end{align*}
\end{lemma}
\begin{proof}
For any $\mathbf{i} \in \mathcal I$ and sufficiently large $B$, we have
\begin{align*}
&\mbox{Pr}(\mathbf{Y} \in \mathcal{Y}_\mathbf{i} |\mathbf{i})\\
=\ &\sum_{\mathbf{x}}\mbox{Pr}(\mathbf{Y} \in \mathcal{Y}_\mathbf{i} |\mathbf{x}) p(\mathbf{x}|\mathbf{i})\\
\geq \ &\sum_{\mathbf{x}\in S(X^n,I_n)|_\mathbf{i}}\mbox{Pr}(\mathbf{Y} \in \mathcal{Y}_\mathbf{i} |\mathbf{x}) p(\mathbf{x}|\mathbf{i})\\
\geq \ & \sum_{\mathbf{x}\in S(X^n,I_n)|_\mathbf{i}}\mbox{Pr}(\mathbf{Y} \in \Gamma_{\sqrt{nB} (\sqrt{2Na_n\ln2}+3\sqrt{N\epsilon})  }(\mathcal{Z}_{(\mathbf{x},\mathbf{i})}) |\mathbf{x})  p(\mathbf{x}|\mathbf{i})\\
\geq \ & (1-\epsilon) \mbox{Pr}(S(X^n,I_n)|_\mathbf{i}   |\mathbf i)\\
\geq \ &(1-\epsilon) (1-2\sqrt[8]{\epsilon})\\
\geq \ &1-3\sqrt[8]{\epsilon}.
\end{align*}

Now consider any $\mathbf y \in \mathcal{Y}_\mathbf{i}$. There exists some $\mathbf{x} \in S(X^n,I_n)|_\mathbf{i}$ such that
$\mathbf y \in \Gamma_{\sqrt{nB} (\sqrt{2Na_n\ln2}+3\sqrt{N\epsilon})  }(\mathcal{Z}_{(\mathbf{x},\mathbf{i})})$. It then follows immediately from Part 2) of Lemma \ref{L:Blown-Up} that
\begin{align*}
 f(\mathbf{y}|\mathbf{i})\geq 2^{-nB(b_n-c_n + \frac{1}{2}\log 2\pi eN + a_n+\sqrt{2a_n\ln2}\log e    + \epsilon_3)} . \end{align*}
\end{proof}

Finally, choosing $\delta$ to be $3\sqrt[8]{\epsilon}$ completes the proof of Lemma \ref{L:Keylemma}.

\subsection{Extension to The $N_1\neq N_2$ Case}

We now prove Theorem \ref{T:newboundsym} for the general case when $N_1$ and $N_2$ are not the same. Note that \dref{E:newboundsym1}--\dref{E:newboundsym2} in the theorem follow immediately along the same lines as the proofs of \dref{E: symnew1}--\dref{E: symnew2}, i.e., by applying Fano's inequality and letting $H(I_n|X^n)=na_n$, so in the sequel we only prove \dref{E:newboundsym3}.

First consider the case of $N_1\leq N_2$. In this case we can equivalently think of  $Z$ and $Y$ as given by
\begin{numcases}{}
Z=X+W_1\nonumber \\
Y=X+W_{21} +W_{22}\nonumber
\end{numcases}
where $W_1, W_{21}, W_{22}$ are zero-mean Gaussian random variables with variances $N_1,N_1,N_2-N_1$ respectively, and they are independent of each other and $X$. Based on this, we write
\begin{numcases}{}
Z^n=X^n+W^n_1\label{E:de1} \\
Y^n=\tilde Z^n+W_{22}^n\label{E:de2}
\end{numcases}
where
\begin{align}
\tilde Z^n:= X^n+W_{21}^n. \label{E:de3}
\end{align}
To prove \dref{E:newboundsym3} we continue with \eqref{eq:intN1<N2} and modify the proof for the symmetric case to be:
\begin{align*}
nR&\leq I (X^n;Y^n,I_n)+n\mu \nonumber  \\
&\leq I (X^n;\tilde{Z}^n,I_n)+n\mu \\
&= I(X^n;\tilde{Z}^n)+I(X^n;I_n)-I(\tilde{Z}^n;I_n)+ n\mu,
\end{align*}
where the second inequality follows from the data processing inequality applied to the Markov chain $ X^n - (\tilde Z^n, I_n) - (  Y^n, I_n)$. Now observe that $(I_n,Z^n, X^n, \tilde Z^n)$ satisfy the conditions of Lemma~\ref{eq:upperboundlemma} and therefore we have
\begin{align*}
nR & \leq nI(X_Q; \tilde Z_Q)+ n( a_n+\sqrt{2a_n\ln2}\log e  )+n\mu\\
&=n(I(X_Q;  Z_Q)+ a_n+\sqrt{2a_n\ln2}\log e  +\mu),
\end{align*}
where $a_n=\frac{1}{n}H(I_n|X^n)$. This proves constraint \dref{E:newboundsym3} for the $N_1\leq N_2$ case.

\vspace{2mm}

Now assume $N_1\geq N_2$. Construct an auxiliary random variable $\tilde Z^n$ as
$$
\tilde Z^n := Y^n+\tilde W^n,
$$
where $\tilde W^n$ is an i.i.d. sequence of Gaussian random variables with zero mean and variance $N_1-N_2$, and is independent of the other random variables in the problem. Applying Lemma~\ref{eq:upperboundlemma} to $(I_n, Z^n, X^n, \tilde Z^n)$ we have
\begin{equation*}
 I (X^n; I_n) - I(\tilde Z^n;I_n)  \leq n( a_n+\sqrt{2a_n\ln2}\log e  ),
\end{equation*}
which combined with the Markov relation $I_n - X^n - Y^n - \tilde Z^n$ further implies that
\begin{equation}
 I (X^n; I_n) - I(Y^n;I_n)  \leq n( a_n+\sqrt{2a_n\ln2}\log e  ). \label{E:concludeplug1}
\end{equation}
Combining this with inequality \eqref{eq:int2} then proves constraint \dref{E:newboundsym3} for the $N_1\geq N_2$ case and cocludes the proof of Theorem \ref{T:newboundsym}.

\section{Further Improvement}\label{S:furtherimprove}
In this section we show that  in the case of $N_1\leq N_2$, our bound in Theorem \ref{T:newboundsym} can be further sharpened for certain regimes of channel parameters. In particular, we will prove the following proposition.

\begin{proposition}\label{T:newbound}
For a Gaussian primitive relay channel with $N_1\leq N_2$, if a rate $R$ is achievable, then there exists some $a\in [0,R_0]$ such that \dref{E: symnew1}, \dref{E: symnew2} and the following two constraints
\begin{numcases}{}
R    \leq   \frac{1}{2}\log \left(1+ \frac{P}{N_1} \right)+a+\sqrt{2a\ln2}\log e   \label{E:sharpen1}\\
R    \leq   \frac{1}{2}\log \left(1+\frac{P}{N_1}\right)  +   \frac{N_1}{N_2}a  +
\sqrt{ \frac{N_1}{N_2}\left(\frac{N_1}{N_2} 2a   \ln2+ 1- \frac{N_1}{N_2}  \right)}\log e  \label{E:sharpen2}
\end{numcases}
are simultaneously satisfied.


\end{proposition}

Proposition \ref{T:newbound} improves upon Theorem \ref{T:newboundsym} for the $N_1\leq N_2$ case by introducing a new constraint \dref{E:sharpen2} that is structurally similar to \dref{E:sharpen1}.
Note that neither  constraint \dref{E:sharpen1} nor \dref{E:sharpen2} is dominating the other and which one is tighter depends on the channel parameter. This makes the bound in Proposition \ref{T:newbound} in general tighter than that in Theorem \ref{T:newboundsym} for the $N_1\leq N_2$ case. Nevertheless, in Appendix \ref{A:largestgapGeneral} we show that the largest gap between the bound in Proposition \ref{T:newbound} and the cut-set bound remains to be $0.0535$, which is still attained when $\frac{P}{N_1}=\frac{P}{N_2}\to \infty$ and $R_0=0.5$.

To show Proposition \ref{T:newbound}, we only need to show the new constraint \dref{E:sharpen2}. In the sequel we will first give a sketch to illustrate the main argument for proving \dref{E:sharpen2} and then present the rigorous proof.

\subsection{Main Argument for Proving \dref{E:sharpen2} }\label{S:Generalargument}
To prove \dref{E:sharpen2} we will prove the following new upper bound on  $I (X^n; I_n) - I(Y^n;I_n)$ in the $N_1\leq N_2$ case:\begin{equation}
 I (X^n; I_n) - I(Y^n;I_n)  \leq  I(X^n;Z^n) - I(X^n;Y^n) +   n\left(\frac{N_1}{N_2}a_n  +
\sqrt{ \frac{N_1}{N_2}\left(\frac{N_1}{N_2} 2a_n \ln2+ 1- \frac{N_1}{N_2}  \right)}\log e\right). \label{E:newconstMI}
\end{equation}
In particular, we again look at the random variables $(I_n, Z^n, X^n, \tilde Z^n, Y^n)$ as specified in \dref{E:de1}--\dref{E:de3},
and their $B$-letter i.i.d. extensions. Since $(I_n, Z^n, X^n, \tilde Z^n)$ satisfy the conditions of Lemma~\ref{eq:upperboundlemma}, from the proof of Lemma~\ref{eq:upperboundlemma}  (c.f. Section \ref{SS:outlinesym} in particular) we have for a typical $(\tilde{\mathbf z},\mathbf i)$ pair,  there exists some $\mathbf{z}$ belonging to the $\mathbf i$th bin such that
\begin{align}
d(\tilde{\mathbf z}, \mathbf{z})\leq  \sqrt{nB}\sqrt{2N_1a_n\ln2}.\label{E: pyta1}
\end{align}
Moreover, since $\mathbf Y=\tilde{\mathbf Z}+\mathbf{W}_{22}$ with $\tilde{\mathbf Z}$ and $\mathbf{W}_{22}$ being independent, it can be shown that for a fixed pair of $(\tilde{\mathbf z}, {\mathbf z})$, the following pythagorean relation holds with high probability:
\begin{align*}
d^2( {\mathbf Y}, \mathbf{z})   & \approx d^2( \mathbf{Y},\tilde{\mathbf z} ) + d^2(\tilde{\mathbf z}, \mathbf{z})  \\
& \approx nB(N_2-N_1) + d^2(\tilde{\mathbf z}, \mathbf{z}).
\end{align*}
This fact combined with \dref{E: pyta1} yields that  for a typical $(\mathbf{y},\mathbf{i})$ pair, there exists some $\mathbf{z}$ belonging to the $\mathbf i$th bin such that
\begin{align}
d(\mathbf{y}, \mathbf{z})&\leq  \sqrt{ nB(N_2-N_1)  +   nB 2N_1a_n\ln2 }\nonumber \\
 &=   \sqrt{nB} \sqrt{N_2    + N_1(2a_n\ln2-1)} . \label{E: geochange2}
\end{align}
%

We now lower bound the conditional density $f(\mathbf y|\mathbf i)$ for a typical $(\mathbf y, \mathbf i)$ pair
based on the geometric relation \dref{E: geochange2}. Similarly as in Section \ref{SS:outlinesym}, we consider the set of  $\mathbf x$'s that are jointly typical with the $\mathbf{z}$ satisfying \dref{E: geochange2}. Again it can be shown that the $\mathbf x$'s that are jointly typical with this $\mathbf z$ satisfy $$p(\mathbf x|\mathbf i)\doteq 2^{-BH(X^n|I_n)}$$ and
$$d(\mathbf x, \mathbf z)\leq \sqrt{nBN_1}.$$
Therefore, by the triangle inequality for each $\mathbf x$ in this set
\begin{align*}
d(\mathbf x, \mathbf y)& \leq \sqrt{nB}(\sqrt{N_1}+\sqrt{N_2    + N_1(2a_n\ln2-1)} ),
\end{align*}
which leads to the following lower bound on $f(\mathbf y|\mathbf x)$,
\begin{align}f(\mathbf y|\mathbf x) &\stackrel {.}{\geq}  2^{ -nB \left(\frac{1}{2}\log 2\pi e N_2 + \frac{N_1}{N_2}a_n  +
\sqrt{ \frac{N_1}{N_2}\left(\frac{N_1}{N_2} 2a_n \ln2+ 1- \frac{N_1}{N_2}  \right)}\log e \right) }   \label{E: fyx}
\end{align}
by using the fact that $\mathbf Y$ is Gaussian given $\mathbf x$. Since the set of such $\mathbf x$'s has cardinality approximately given by $2^{BH(X^n|Z^n)}$, we have
\begin{align*}
f(\mathbf{y}|\mathbf{i})&=\sum_{\mathbf x} f(\mathbf{y}|\mathbf{x})p(\mathbf{x}|\mathbf{i})\nonumber \\
&\stackrel {.}{\geq}  2^{B(H(X^n|Z^n)-H(X^n|I_n) )}\times \text{R.H.S. of \dref{E: fyx}}\nonumber\\
&=  2^{-B\left(H(X^n|I_n)-H(X^n|Z^n)+ \frac{n}{2}\log 2\pi e N_2 + n\left(\frac{N_1}{N_2}a_n  +
\sqrt{ \frac{N_1}{N_2}\left(\frac{N_1}{N_2} 2a_n \ln2+ 1- \frac{N_1}{N_2}  \right)}\log e\right)    \right)         }.
\end{align*}
Finally, translating the above lower bound on $f(\mathbf{y}|\mathbf{i})$ for typical $(\mathbf{y},\mathbf{i})$ pairs to the upper bound on $h(Y^n|I_n)$, we have
\begin{align}
h(Y^n|I_n) &\leq  H(X^n|I_n)-H(X^n|Z^n)+ \frac{n}{2}\log 2\pi e N_2 + n\left(\frac{N_1}{N_2}a_n  +
\sqrt{ \frac{N_1}{N_2}\left(\frac{N_1}{N_2} 2a_n \ln2+ 1-\frac{N_1}{N_2}  \right)}\log e \right) . \label{E:desire}
\end{align}
Then the new bound \dref{E:newconstMI} on $I (X^n; I_n) - I(Y^n;I_n)$ can be proved as follows:
\begin{align}
I (X^n; I_n) - I(Y^n;I_n) & = H(X^n)-H(X^n|I_n) - h(Y^n)  +  h(Y^n|I_n) \label{E:similarmanr1} \\
& \leq  H(X^n)-H(X^n|I_n) -  h(Y^n) + H(X^n|I_n)-H(X^n|Z^n)   \\
&~~~  + \frac{n}{2}\log 2\pi e N_2 + n\left(\frac{N_1}{N_2}a_n  +
\sqrt{ \frac{N_1}{N_2}\left(\frac{N_1}{N_2} 2a_n \ln2+ 1-\frac{N_1}{N_2}  \right)}\log e \right)   \\
& =  I(X^n;Z^n) - [ h(Y^n) - \frac{n}{2}\log 2\pi eN_2]\\
&~~~ +  n\left(\frac{N_1}{N_2}a_n  +
\sqrt{ \frac{N_1}{N_2}\left(\frac{N_1}{N_2} 2a_n \ln2+ 1-\frac{N_1}{N_2}  \right)}\log e \right)  \\
& =  I(X^n;Z^n) - I(X^n;Y^n) +  n\left(\frac{N_1}{N_2}a_n  +
\sqrt{ \frac{N_1}{N_2}\left(\frac{N_1}{N_2} 2a_n \ln2+ 1-\frac{N_1}{N_2}  \right)}\log e \right).
\end{align}
With this we conclude that for any achievable rate $R$, any $\mu >0$ and $n$ sufficiently large,
\begin{align}
nR&\leq  I(X^n;Y^n)+ I (X^n; I_n) - I(Y^n;I_n) +n\mu \\
&\leq nI(X_Q;Z_Q) +n\left(\frac{N_1}{N_2}a_n  +
\sqrt{ \frac{N_1}{N_2}\left(\frac{N_1}{N_2} 2a_n \ln2+ 1- \frac{N_1}{N_2}  \right)}\log e\right)+n\mu,\label{E:similarmanr2}
\end{align}
which is essentially constraint \dref{E:sharpen2}.

\subsection{Formal Proof of \dref{E:sharpen2} }
We now formalize the argument and give a rigorous proof of \dref{E:sharpen2}.  We shall adopt the same definitions and notations of $\tilde S(X^n,Z^n)$, $S(X^n,Z^n)$, $S(X^n,Z^n)$,  $S(X^n,I_n)$ and $\mathcal{Z}_{(\mathbf{x},\mathbf{i})}$ as in the symmetric case treated in Section \ref{SS:ProoftoKeylemma}. Since the relations among $(\mathbf X, \mathbf Z, \mathbf I)$ remain unchanged compared to the symmetric case, Lemmas \ref{L:unchange1}--\ref{L: properties_S(x,i)} will still apply here. Also because now $\mathbf{Z}$ and $\tilde{\mathbf{Z}}$ are identically distributed given $\mathbf X$, by Lemma \ref{L:Blown-Up}-1), we have for any $(\mathbf{x},\mathbf{i}) \in S(X^n,I_n)$ and sufficiently large $B$,
\begin{align*}
\mbox{Pr}(\tilde{\mathbf{Z}}  \in \Gamma_{\sqrt{nB}  (\sqrt{2N_1a_n}+3\sqrt{N_1\epsilon}) }(\mathcal{Z}_{(\mathbf{x},\mathbf{i})}) |\mathbf{x}) \geq 1-\epsilon ,
\end{align*}
i.e.,
\begin{align}
\mbox{Pr}( \exists ~\mathbf{z}\in  \mathcal{Z}_{(\mathbf{x},\mathbf{i})} \text{~s.t.~} d ( \tilde{\mathbf{Z}},\mathbf{z})\leq \sqrt{nB}  (\sqrt{2N_1a_n\ln2}+3\sqrt{N_1\epsilon})   |\mathbf{x}) \geq 1-\epsilon. \label{E:view1}
\end{align}


Now consider any specific pair of $(\tilde{\mathbf{z}},\mathbf{z})$ with $d ( \tilde{\mathbf{z}},\mathbf{z})\leq \sqrt{nB}  (\sqrt{2N_1a_n\ln2}+3\sqrt{N_1\epsilon}) $ and $\mathbf Y=\tilde{\mathbf{z}}+\mathbf{W}_{22}$ with $\mathbf{W}_{22}$ being $n$ i.i.d. Gaussian random variables that are independent of $\tilde{\mathbf Z}$ and with zero mean and variance $N_2-N_1$.
We have
\begin{align*}
d^2(\mathbf Y, \mathbf{z})
&= \|\mathbf Y-\mathbf{z} \|^2\\
&=\|\mathbf{W}_{22}+\tilde{\mathbf{z}}-\mathbf{z} \|^2\\
&=[\mathbf{W}_{22}+(\tilde{\mathbf{z}}-\mathbf{z})]^T [\mathbf{W}_{22}+(\tilde{\mathbf{z}}-\mathbf{z})]\\
&= \|\mathbf{W}_{22}\|^2 + 2 \mathbf{W}_{22}^T (\tilde{\mathbf{z}}-\mathbf{z}) + \|(\tilde{\mathbf{z}}-\mathbf{z})\|^2\\
&= \|\mathbf{W}_{22}\|^2 + 2 \mathbf{W}_{22}^T (\tilde{\mathbf{z}}-\mathbf{z}) + d^2(\tilde{\mathbf{z}},\mathbf{z}).
\end{align*}
From the weak law of large numbers, for any $\epsilon >0$ and sufficiently large $B$, we have
\begin{align*}
\mbox{Pr}(\|\mathbf{W}_{22}\|^2 \in [nB(N_2-N_1-\epsilon/2 ), nB(N_2-N_1+\epsilon/2 ) ]  ) \geq 1-\epsilon/2
\end{align*}
and
\begin{align*}
\mbox{Pr}(2 \mathbf{W}_{22}^T (\tilde{\mathbf{z}}-\mathbf{z}) \in [-nB \epsilon/2 , nB \epsilon/2  ]  ) \geq 1-\epsilon/2.
\end{align*}
Therefore, by the union bound, for any $\epsilon >0$ and sufficiently large $B$,
\begin{align}
1-\epsilon& \leq \mbox{Pr}(d^2(\mathbf Y, \mathbf{z}) \leq nB(N_2-N_1+\epsilon )+d^2(\tilde{\mathbf{z}},\mathbf{z})    ) \nonumber \\
& \leq \mbox{Pr}(d^2(\mathbf Y, \mathbf{z}) \leq nB(N_2-N_1+\epsilon )+  nB(\sqrt{2N_1a_n\ln2}+3\sqrt{N_1\epsilon})^2         ) \nonumber \\
&= \mbox{Pr}(d(\mathbf Y, \mathbf{z}) \leq \sqrt{nB}\sqrt{(N_2+N_1 (2 a_n\ln2 -1) +\epsilon_1      )} ), \label{E:view2}
\end{align}
where $\epsilon_1$ is defined such that
$$ (N_2-N_1+\epsilon )  +(\sqrt{2N_1a_n\ln2}+3\sqrt{N_1\epsilon})^2   = N_2+N_1(2a_n\ln2-1)+ \epsilon_1$$
and $\epsilon_1\to 0$ as $\epsilon \to 0$. In light of \dref{E:view1} and \dref{E:view2}, we have for sufficiently large $B$,
\begin{align*}
& \mbox{Pr}( {\mathbf{Y}}  \in \Gamma_{\sqrt{nB}  \sqrt{N_2+N_1(2a_n\ln2-1)+ \epsilon_1}   }(\mathcal{Z}_{(\mathbf{x},\mathbf{i})}) |\mathbf{x})\\
= \ &\mbox{Pr}( \exists ~\mathbf{z}\in  \mathcal{Z}_{(\mathbf{x},\mathbf{i})} \text{~s.t.~} d ( {\mathbf{Y}},\mathbf{z})\leq \sqrt{nB}  \sqrt{N_2+N_1(2a_n\ln2-1)+ \epsilon_1}     \big|\mathbf{x}) \\
\geq \ & \mbox{Pr}( \exists ~\mathbf{z}\in  \mathcal{Z}_{(\mathbf{x},\mathbf{i})} \text{~s.t.~} d ( \tilde{\mathbf{Z}},\mathbf{z})\leq \sqrt{nB}  (\sqrt{2N_1a_n\ln2}+3\sqrt{N_1\epsilon})   \big|\mathbf{x})  \\
&\times \mbox{Pr}( \exists ~\mathbf{z}\in  \mathcal{Z}_{(\mathbf{x},\mathbf{i})} \text{~s.t.~} d ( {\mathbf{Y}},\mathbf{z})\leq \sqrt{nB}  \sqrt{N_2+N_1(2a_n\ln2-1)+ \epsilon_1}    \\
&~~~~~~~~~~~~~~~~~~~~~~~~~~~~~~~~~~~~~~\big|\mathbf{x},\exists ~\mathbf{z}\in  \mathcal{Z}_{(\mathbf{x},\mathbf{i})} \text{~s.t.~} d ( \tilde{\mathbf{Z}},\mathbf{z})\leq \sqrt{nB}  (\sqrt{2N_1a_n\ln2}+3\sqrt{N_1\epsilon}) )  \\
\geq \ & (1-\epsilon)^2.
\end{align*}
Using this fact and following the lines to prove Lemma \ref{L:Blown-Up}-2), it can be shown that for any $\mathbf{y}\in \Gamma_{\sqrt{nB}  \sqrt{N_2+N_1(2a_n\ln2-1)+ \epsilon_1}   }(\mathcal{Z}_{(\mathbf{x},\mathbf{i})})$ with $(\mathbf{x},\mathbf{i}) \in S(X^n,I_n)$, the following lower bound on $f(\mathbf{y}|\mathbf{i})$ holds:
\begin{align}
f(\mathbf{y}|\mathbf{i})\geq 2^{-B\left(H(X^n|I_n)-H(X^n|Z^n)+ \frac{n}{2}\log 2\pi e N_2 + n\left(\frac{N_1}{N_2}a_n  +
\sqrt{ \frac{N_1}{N_2}\left(\frac{N_1}{N_2} 2a_n \ln2+ 1-\frac{N_1}{N_2}  \right)}\log e\right)  +n\epsilon'  \right)         } \label{E:trans}
\end{align}
where $\epsilon' \to 0$ as $\epsilon \to 0$ and $B\to \infty$. Finally, following the same procedure as in the symmetric case to translate \dref{E:trans} into
the desired entropy relation \dref{E:desire} and then using that in the manner of \dref{E:similarmanr1}--\dref{E:similarmanr2} prove constraint \dref{E:sharpen2} and  Proposition \ref{T:newbound}.

\appendices

\section{Proof of Corollary \ref{C:cutset}}\label{Gaussianoptimal}
In this appendix, we show that both of the mutual informations $I(X;Y,Z)$ and $I(X;Y)$ in Proposition \ref{P:cutset} are maximized when $X\sim\mathcal N (0,P)$ and thus establish Corollary \ref{C:cutset}.

Specifically, consider the following chain of inequalities:
\begin{align*}
 I(X;Y,Z)  &= h(Y,Z)-h(Y,Z|X)\\
&=h(X+W_1,X+W_2)-h(W_1,W_2)\\
&\leq \frac{1}{2}\log (2\pi e)^2 |\Sigma_{(X+W_1,X+W_2)}|-h(W_1)-h(W_2)   \\
&= \frac{1}{2}\log (2\pi e)^2 \left|\begin{matrix}
Var(X)+N_1 &  E[X^2]  \\
E[X^2]   & Var(X)+N_2
\end{matrix}\right|  -\frac{1}{2}\log (2\pi e )^2N_1N_2\\
&\leq  \frac{1}{2}\log [(2\pi e)^2 (N_1N_2+ (N_1+N_2) Var(X)    ) ]         -\frac{1}{2}\log (2\pi e )^2N_1N_2  \\
&= \frac{1}{2}\log \left(1+\frac{(N_1+N_2)Var(X) }{N_1N_2}  \right)\\
&\leq \frac{1}{2}\log \left(1+\frac{P}{N_1}+\frac{P}{N_2}\right)
\end{align*}
where all the inequalities hold with equality when $X\sim \mathcal{N} (0,P)$. Similarly, we have \begin{align*}
  I(X;Y)  \leq \frac{1}{2}\log \left(1+\frac{P}{N_2}\right)
\end{align*}
with the inequality holding with equality when $X\sim \mathcal{N} (0,P)$. Combining the above establishes Corollary \ref{C:cutset}.

\section{Proof of Proposition \ref{P:largestpossible}}\label{A:largestgap}
First rewrite our new bound in Theorem \ref{T:newboundsym} as
\begin{numcases}{}
 R   \leq \frac{1}{2}\log \left(1+\frac{P}{N_1}+\frac{P}{N_2}\right)\label{E:Anewboundsym1}  \\
R    \leq   \frac{1}{2}\log \left(1+\frac{P}{N_2}\right)+R_0  -a^*  \label{E:Anewboundsym2}
\end{numcases}
where $a^*$ is the solution to the following equation:
 \begin{align} \frac{1}{2}\log \left(1+\frac{P}{N_2}\right)+R_0 - \frac{1}{2}\log \left(1+\max \left\{\frac{P}{N_1},\frac{P}{N_2}\right\}\right)=2a^*+\sqrt{2a^*\ln2}\log e. \label{E:Anewboundsym3}
\end{align}
Observe that the gap $\Delta(\frac{P}{N_1},\frac{P}{N_2}, R_0)$ between our new bound and the cut-set bound is positive only if the channel parameters $(\frac{P}{N_1},\frac{P}{N_2}, R_0)$ are such that between \dref{E:Anewboundsym1} and \dref{E:Anewboundsym2} of our bound, constraint \dref{E:Anewboundsym2} is active. This is because if in our bound constraint \dref{E:Anewboundsym1}  is active, then for the cut-set bound also \dref{E:cutsetgeneral1} is active and these two bounds become the same.

Thus to find the largest gap, one can without loss of generality assume constraint \dref{E:Anewboundsym2} is active for our bound. We now argue that the largest gap happens only when \dref{E:cutsetgeneral2} is active for the cut-set bound. Suppose this is not true, i.e., when the largest gap happens constraint \dref{E:cutsetgeneral1} instead of \dref{E:cutsetgeneral2} is active. Then this implies that the R.H.S. of \dref{E:cutsetgeneral1} is strictly less than that of \dref{E:cutsetgeneral2} and thus one can reduce $R_0$ to further increase the gap, which contradicts with the largest gap assumption. Therefore, only when \dref{E:Anewboundsym2} and \dref{E:cutsetgeneral2} are active, the gap attains the largest value that is given by the solution $a^*$ to equation \dref{E:Anewboundsym3}. The largest value that the L.H.S. of \dref{E:Anewboundsym3} can take while still maintaining \dref{E:Anewboundsym2} and \dref{E:cutsetgeneral2} are active is $0.5$, in which case the channel parameter $(\frac{P}{N_1},\frac{P}{N_2},R_0)$ has to be $(\infty,\infty, 0.5)$. Solving equation \dref{E:Anewboundsym3}  with $\text{L.H.S.}=0.5$, we obtain $\Delta^* =  \Delta(\infty,\infty, 0.5)=0.0535$.

\section{}\label{A:largestgapGeneral}

Consider the following upper bound jointly imposed by \dref{E: symnew1}--\dref{E: symnew2} and \dref{E:sharpen2},
\begin{numcases}{}
R   \leq \frac{1}{2}\log \left(1+\frac{P}{N_1}+\frac{P}{N_2}\right)  \label{E:Agapcom1'} \\
R    \leq   \frac{1}{2}\log \left(1+\frac{P}{N_2}\right)+R_0-a^* \label{E:Agapcom2'}
\end{numcases}
where $a^*$ is the solution to the equation
\begin{align}
\frac{1}{2}\log \left(1+\frac{P}{N_2}\right)+R_0-\frac{1}{2}\log \left(1+\frac{P}{N_1}\right)  = \left(\frac{N_1}{N_2}+1\right)a^*   +\sqrt{ \frac{N_1}{N_2}\left(\frac{N_1}{N_2} 2a^* \ln2+1-   \frac{N_1}{N_2}  \right)}\log e  .     \label{E:Asolvegap'}
\end{align}
To show that the largest gap between our bound in Proposition \ref{T:newbound} and the cut-set bound in \dref{E:cutsetgeneral1}--\dref{E:cutsetgeneral2}  remains to be $\Delta^*=0.0535$, it suffices to show that the above bound and the cut-set bound  differ from each other at most 0.0535.

Similarly as in Appendix \ref{A:largestgap}, one can argue that the largest gap between the above bound and the cut-set bound happens only when  \dref{E:Agapcom2'} and \dref{E:cutsetgeneral2} are active respectively, in which case the gap is given by the $a^*$ satisfying \dref{E:Asolvegap'}. Note that for \dref{E:cutsetgeneral2} to be active in the cut-set bound, one must have
$$\frac{1}{2}\log \left(1+\frac{P}{N_2}\right)+R_0 \leq \frac{1}{2}\log \left(1+\frac{P}{N_1}+\frac{P}{N_2}\right).$$
Then to find the largest $a^*$ we impose the following relation:
\begin{align*}
\frac{1}{2}\log \left(1+\frac{P}{N_1}+\frac{P}{N_2}\right)
-\frac{1}{2}\log \left(1+\frac{P}{N_1}\right)  = \left(\frac{N_1}{N_2}+1\right)a^*   +\sqrt{ \frac{N_1}{N_2}\left(\frac{N_1}{N_2} 2a^* \ln2+1-   \frac{N_1}{N_2}  \right)}\log e  .    
\end{align*}
Letting $x_i=\frac{P}{N_i}$ for $i\in \{1,2\}$ and  solving the above equation, we have
\begin{align*}
a^*&=\frac{ (\frac{x_2}{x_1}+1) \ln (1+\frac{x_2}{1+x_1})+2  \frac{x_2^2}{x_1^2}  }{2(\frac{x_2}{x_1}+1)^2\ln 2}\\
&~~~- \frac{\sqrt{((\frac{x_2}{x_1}+1) \ln (1+\frac{x_2}{1+x_1})+2  \frac{x_2^2}{x_1^2})^2 -  (\frac{x_2}{x_1}+1)^2[\ln^2 (1+\frac{x_2}{1+x_1})
+\frac{x_2}{x_1}(\frac{x_2}{x_1}-1)  ]   }}{2(\frac{x_2}{x_1}+1)^2\ln 2}
\end{align*}
which attains the maximum value $a^*=0.0535$ when $x_1=x_2=\infty$. This shows that the largest gap between our bound in Proposition \ref{T:newbound} and the cut-set bound remains to be 0.0535.

\section{Proof of Lemma \ref{L:Talagrand}}\label{A:ProoftoTalagrand}
Given $A\subseteq \mathbb R^{n}$, let $B:=\{\underline{\omega}\in \mathbb R^{n}  :    \sqrt{N} \underline{\omega} \in A   \}$ and $V_i=\frac{U_i}{\sqrt{N}}, \forall i \in \{1,2,\ldots,n\}$. Then $V_1, V_2, \ldots, V_n$ are $n$ i.i.d. standard Gaussian random variables with $V_i \sim \mathcal{N}(0,1),\forall i \in \{1,2,\ldots,n\}$, and
\begin{align*}
\mbox{Pr}(V^n \in B)= \mbox{Pr}(\sqrt{N} V^n \in A)=\mbox{Pr}(U^n \in A)   \geq 2^{-na_n}.
\end{align*}
We next invoke Gaussian measure concentration as stated in (1.6) of \cite{Talagrand}: for any $B\subseteq \mathbb R^{n}$ and $$t\geq \sqrt{-2\ln \text{Pr}(V^n \in B)},$$ we have
\begin{align*}
\mbox{Pr}(V^n \in \Gamma_{t } (B) )\geq 1-e^{-\frac{1}{2}\left( t- \sqrt{-2\ln \text{Pr}(V^n \in B)}  \right)^2}.
\end{align*}
Thus,  for any $r> 0$,
\begin{align*}
&\mbox{Pr}(V^n \in \Gamma_{\sqrt{n}( \sqrt{2a_n\ln 2 } +\frac{r}{\sqrt{N}} ) } (B) )\\
\geq \ &\mbox{Pr}(V^n \in \Gamma_{\sqrt{-2\ln \text{Pr}(V^n \in B)} + \sqrt{\frac{n}{N}}r    } (B) )\\
\geq \ & 1-2^{-\frac{nr^2}{2N}}.
\end{align*}
Noting that
$$\Gamma_{\sqrt{n}( \sqrt{2Na_n\ln 2 } +r ) } (A)
=\left\{\sqrt{N} \underline{\omega} :     \underline{\omega} \in \Gamma_{\sqrt{n}( \sqrt{2a_n\ln 2} +\frac{r}{\sqrt{N}}) } (B)  \right\},$$
we have
\begin{align*}
&\mbox{Pr}(U^n \in \Gamma_{\sqrt{n}( \sqrt{2Na_n\ln2} +r ) } (A) )\\
=\ &\mbox{Pr}(\sqrt{N}V^n \in \Gamma_{\sqrt{n}( \sqrt{2Na_n\ln2} +r ) } (A) )\\
=\ & \mbox{Pr}(V^n \in \Gamma_{\sqrt{n}( \sqrt{2a_n\ln2} +\frac{r}{\sqrt{N}} ) } (B) )\\
\geq \ & 1-2^{-\frac{nr^2}{2N}}.
\end{align*}

\end{document}